\documentclass[preprint]{imsart}

\RequirePackage[OT1]{fontenc}
\RequirePackage{amsthm,amsmath}
\RequirePackage[numbers]{natbib}
\RequirePackage[colorlinks,citecolor=blue,urlcolor=blue]{hyperref}
\usepackage[normalem]{ulem} 
\usepackage{soul}
\usepackage{ucs}
\usepackage{bbm}
\usepackage[toc]{appendix}
\usepackage[utf8x]{inputenc}
\usepackage{amsfonts}
\usepackage{amssymb}
\usepackage{mathrsfs}  
\usepackage{fontenc}
\usepackage{graphicx}

\startlocaldefs
\newcommand{\ddr}{\mathrm{d}}
\numberwithin{equation}{section}
\theoremstyle{plain}
\newtheorem{lem}{Lemma}[section]
\newtheorem{prp}{Proposition}[section]
\newtheorem{thm}{Theorem}[section]

\allowdisplaybreaks

\def\bfD{\mathbf{D}}

\def\bbP{\mathbb{P}}

\endlocaldefs

\begin{document}

\begin{frontmatter}
\title{Bayesian nonparametric analysis of Kingman's coalescent}
\runtitle{BNP analysis of Kingman's coalescent}

\begin{aug}
\author{\fnms{Stefano} \snm{Favaro}\thanksref{m4}\ead[label=e1]{stefano.favaro@unito.it}},
\author{\fnms{Shui} \snm{Feng}\thanksref{}\ead[label=e2]{shuifeng@mcmaster.ca}},
\and
\author{\fnms{Paul A.} \snm{Jenkins}\thanksref{}\ead[label=e3]{P.Jenkins@warwick.ac.uk}}

\thankstext{m4}{Also affiliated to Collegio Carlo Alberto (Torino, Italy) and IMATI-CNR ``Enrico Magenes" (Milan, Italy).}
\runauthor{S. Favaro, S. Feng, and P. A. Jenkins}

\affiliation{University of Torino\thanksmark{m1}, McMaster University\thanksmark{m2}, and University of Warwick\thanksmark{m3}}

\address{Department of Economics and Statistics, University of Torino\\
Torino 10134, Italy\\
\printead{e1}}

\address{Department of Mathematics and Statistics, McMaster University\\
Hamilton L8S4K1, Canada\\
\printead{e2}}

\address{Department of Statistics \& Department of Computer Science, University of Warwick\\
Coventry CV4 7AL, United Kingdom\\
\printead{e3}}
\end{aug}

\vspace{0.7cm}
\textit{Dedicated to the memory of Paul Joyce}
\vspace{0.5cm}

\begin{abstract}
Kingman's coalescent is one of the most popular models in population genetics. It describes the genealogy of a population whose genetic composition evolves in time according to the Wright-Fisher model, or suitable approximations of it belonging to the broad class of Fleming-Viot processes. Ancestral inference under Kingman's coalescent has had much attention in the literature, both in practical data analysis, and from a theoretical and methodological point of view. Given a sample of individuals taken from the population at time $t>0$, most contributions have aimed at making frequentist or Bayesian parametric inference on quantities related to the genealogy of the sample. In this paper we propose a Bayesian nonparametric predictive approach to ancestral inference. That is, under the prior assumption that the composition of the population evolves in time according to a neutral Fleming-Viot process, and given the information contained in an initial sample of $m$ individuals taken from the population at time $t>0$, we estimate quantities related to the genealogy of an additional unobservable sample of size $m^{\prime}\geq1$. As a by-product of our analysis we introduce a class of Bayesian nonparametric estimators (predictors) which can be thought of as Good-Turing type estimators for ancestral inference. The proposed approach is illustrated through an application to genetic data.
\end{abstract}



\end{frontmatter}


\section{Introduction}

The Wright-Fisher (WF) model is a popular discrete-time model for the evolution of gene frequencies in a population. Consider a population of individuals, i.e. chromosomes, and assume that  each individual has an associated genetic type, with $\mathcal{X}$ being the set of possible types. In the classical WF model the population has constant (large) size $N$ and it evolves in discrete non-overlapping generations according to the following random processes: i) each individual in the next generation chooses, uniformly at random, an individual in the current generation and copies it, with the choice made by different individuals being independent; ii) the type of each progeny of an individual of type $i\in \mathcal{X}$ is $i$ with probability $1-\delta$, and $j$ with probability $\delta p_{i,j}$, that is mutations occur with probability $\delta\in(0,1)$ per individual per generation, according to a Markov chain with zero-diagonal transition matrix $P=(p_{i,j})_{i\geq1,j\geq1}$. Two additional common assumptions are that $P$ has a unique stationary distribution and that the evolution of the population is neutral, namely all variants in the population are equally fit and are thus equally likely to be transmitted. The assumption of neutrality allows for a crucial simplification of the above WF evolution. Indeed under this assumption the random process describing the demography of the population becomes independent of the random process describing the genetic types carried by the individuals. Although rather simple, the neutral WF model captures many important features of the evolution of human and other populations, thus providing a statistical model which is at the basis of most existing inference methods in population genetics. We refer to the monographs by Ewens \cite{Ewe(04)} and Tavar\'e \cite{Tav(04)} for a comprehensive and stimulating account on the WF model.

In the WF model one can describe the genetic composition of the population at any point in time by giving a list of the genetic types currently present, and the corresponding proportion of the population currently of each type. Note that such a description corresponds to giving a probability measure on the set $\mathcal{X}$ of possible population types. In such a framework one obtains a discrete time probability-measure-valued Markov process, namely a discrete time Markov process whose state space corresponds to the space of the probability measures on $\mathcal{X}$. As the population size $N$ becomes large a suitable rescaling of the Markov process converges to a diffusion limit: time is measured in units of $N$ generations, and the mutation rates are rescaled as $N^{-1}$. The limiting process, called the Fleming-Viot (FV) process, is formulated as a diffusion process whose state space is the space of probability measures on an arbitrary compact metric space $\mathcal{X}$. See Ethier and Kurtz \cite{Eth(93a)} and references therein for a rigorous treatment with a view towards population genetics. Intuitively, the FV process can thus be thought of as an approximation to a large population evolving in time according to the WF model. For instance, the classical WF diffusion on the set $[0,1]$ is a special case of the FV process which arises when there are only two possible genetic types and one tracks the population frequency of one of the types. 

The coalescent arises by looking backward in time at the evolution described by the WF model. See, e.g., the seminal works by Griffiths \cite{Gri(80)}, Kingman \cite{Kin(82a)}, Kingman \cite{Kin(82b)} and Tavar\'e \cite{Tav(84)}, and the monographs by Ewens \cite{Ewe(04)} and Tavar\'e \cite{Tav(04)}. Consider a population evolving in time as a WF model with scaled mutation rate $\alpha=\theta/(2N)$, for $\theta>0$, and with parent-independent transition matrix $P$. In the large $N$ population limit the genealogical history of a sample of $m$ individuals from the population may be described by a rooted random binary tree, where branches represent genealogical lineages, i.e., lines of descent. The tree initially has $m$ lineages for a period of time $T_{m}$, after which a lineage is lost for the occurrence of one of the following events: i) a mutation according to the transition matrix $P$; ii) a coalescence, namely a pair of lineages, chosen uniformly at random and independently of all other events, join. Recursively, the times $T_{k}$, for $k=m,m-1,\ldots,2$ for which the tree has $k$ lineages are independent exponential random variables with parameter $2^{-1}k(k-1+\theta)$, after which a  lineage is lost by mutation with probability $k\theta/(k(k-1+\theta))$ or by coalescence with probability $k(k-1)/(k(k-1+\theta))$. When $\theta = 0$, the resulting random tree is referred to as the $m$-coalescent, or simply the coalescent, and was first described in the seminal work of Kingman \cite{Kin(82a)}. When $\theta >0$ the process is a coalescent with mutation, with antecedents including Griffiths \cite{Gri(80)}. As shown in Donnelly and Kurtz \cite{Don(96)}, in the large $N$ population limit the $m$-coalescent describes the genealogy of a sample of $m$ individuals from a population evolving as the FV process. There exists even a natural limit, as the sample size $m\rightarrow+\infty$, of the $m$-coalescent. This can be thought of as the limit of the genealogy of the whole population, or alternatively as the genealogy of the infinite population described by the FV process.

This paper considers the problem of making ancestral inference, i.e. inference on the genealogy of a genetic population, from a Bayesian nonparametric predictive perspective. The statistical setting we deal with can be described as follows. We consider a population with an (ideally) infinite number of types, and we assume that the population's composition evolves in time according to a neutral FV process whose unique stationary distribution is the law of the Dirichlet process by Ferguson \cite{Fer(73)}. From a Bayesian perspective, the law of the FV process, or its dual law determined with respect to the Kingman's coalescent process, plays the role of a nonparametric prior for the evolution of the population. Given the observed data, which are assumed to be a random sample sample of $m$ individuals from the population at time $t>0$, we characterize the posterior distribution of some statistics of the enlarged $(m+m^{\prime})$-coalescent induced by an additional unobservable sample of size $m^{\prime}\geq1$. Corresponding Bayesian nonparametric estimators, with respect to a squared loss function, are then given in terms of posterior expectations. Of special interest is the posterior distribution of the number of non-mutant lineages surviving from time $0$ to time $t$, that is the number of non-mutant ancestors in generation $t$ in a sample at time $0$.  This, in turn, leads to the posterior distribution of the time of the most recent common ancestor in the $(m+m^{\prime})$-coalescent.  As a by-product of our posterior characterizations we introduce a class of Bayesian nonparametric estimators of the probability of discovery of non-mutant lineages. This is a novel class of estimators which can be thought of as ancestral counterparts of the celebrated Good-Turing type estimators developed in Good \cite{Goo(53)} and Good and Toulmin \cite{Goo(56)}.

Ancestral inference has had much attention in the statistical literature, both in practical data analysis, and from a theoretical and methodological point of view. See, e.g., Griffiths and Tavar\'e \cite{Gri(94)}, Griffiths and Tavar\'e \cite{Gri(94a)}, Stephens and Donnelly \cite{Ste(00)}, Stephens \cite{Ste(01)} and Griffiths and Tavar\'e \cite{Gri(03)}. Given an (observable) random sample sample of $m$ individuals from the population at time $t>0$, most contributions in the literature have aimed at making frequentist or Bayesian inference on quantities related to the genealogy of the sample, e.g., the number of non-mutant lineages, the age of the alleles in the sample, the time of the most recent common ancestor, the age of particular mutations in the ancestry, etc. This is typically done in a parametric setting by using suitable summary statistics of the data, or by combining the full data with suitable approximations of the likelihood function obtained via importance sampling or Markov chain Monte Carlo techniques. In this paper, instead, we introduce a Bayesian nonparametric predictive approach that makes use of the observed sample of $m$ individuals to infer quantities related to the genealogy of an additional unobservable sample. For instance, how many non-mutant lineages would I expect a time $t$ ago if I enlarged my initial observable sample by $m^{\prime}$ unobservable samples? How many of these non-mutant lineages have small frequencies? In the context of ancestral inference, these questions are of great interest because they relate directly to the speed of evolution via the rate of turnover of alleles. See Stephens and Donnelly \cite{Ste(00)} and references therein for a comprehensive discussion. Our approach answers these and other questions under the (prior) assumption that the genealogy of the population follows the Kingman coalescent. To the best of our knowledge, this is the first predictive approach to ancestral inference in this setting. 

The paper is structured as follows. In Section 2 we recall some preliminaries on the neutral FV process and Kingman's coalescent, we introduce new results on ancestral distributions, and we characterize the posterior distribution of the number of non-mutant lineages at time $t$ back in the enlarged $(m+m^{\prime})$-coalescent. A suitable refinement of this posterior distribution and a class of Good-Turing type estimators for ancestral inference are also introduced. In Section 3 we show how to implement our results, and we present a numerical illustration based on genetic data. Section 4 contains a discussion of the proposed methodology, and outlines future research directions. Proofs are deferred to the Appendix.


\section{Ancestral posterior distributions}\label{sec2}
All the random elements introduced in this section are meant to be assigned on a probability space $(\Omega,\mathscr{F},\mathbb{P})$, unless otherwise stated. Let $\mathcal{X}$ be a compact metric space. For any $\theta>0$ and any non-atomic probability measure $\nu_{0}$ on $\mathcal{X}$, let $\Pi(\theta\nu_{0})$ be the distribution of a Dirichlet process on $\mathcal{X}$ with base measure $\theta\nu_{0}$. We refer to Ferguson \cite{Fer(73)} for a  definition and distributional results on the Dirichlet process. In our context $\theta$ will correspond to the mutation rate, $\nu_0$ to the stationary distribution of the mutation process, and $\Pi(\theta\nu_0)$ to the stationary distribution of the population type frequencies in the diffusion limit. For any $n\geq0$ let $d_{n}(t)=\mathbb{P}[D(t)=n]$ where $\{D(t):t\geq0\}$ is a pure death process, $D(0)=+\infty$ almost surely, with rate $\lambda_{n}=2^{-1}n(n-1+\theta)$. It is known from Griffiths \cite{Gri(80)} that
\begin{equation}\label{eq_lin_pop}
d_{n}(t)=(-1)^{n}\sum_{i\geq n}\rho_{i}(t)\frac{{i\choose n}(n+\theta)_{(i-1)}}{i!},
\end{equation}
where 
\begin{displaymath}
\rho_{i}(t)=(-1)^{i}(2i-1+\theta)\text{e}^{-\lambda_{i}t}
\end{displaymath}
for each $t>0$. Here and elsewhere, for any nonnegative $x$ we use $x_{(0)}=x_{[0]}=1$ and, for any $n\geq1$, $x_{(n)}=x(x+1)\cdots(x+n-1)$ and $x_{[n]}=x(x-1)\cdots(x-n+1)$, i.e. rising and falling factorial numbers. If $\alpha=\theta/(2N)$, with $\alpha$ being the mutation rate of the WF model, then $D(t)$ is the number of non-mutant lineages surviving from time $0$ to time $t>0$ in the large $N$ population limit of the WF model when the sample size $m \to\infty$, and $\lambda_n$ is the total backwards-in-time rate of loss of lineages when there are currently $n$ lineages. The pure death process $\{D(t):t\geq0\}$ is typically referred to as the ancestral (genealogical) process. See, e.g., Griffiths \cite{Gri(80)} and Tavar\'e \cite{Tav(84)} for a detailed account on the ancestral process.

Let $\mathcal{P}_{\mathcal{X}}$ be the space of probability measures on $\mathcal{X}$ equipped with the topology of weak convergence. The neutral FV process is a diffusion process on $\mathcal{P}_{\mathcal{X}}$, namely a probability-measure-valued diffusion. Here we focus on the neutral FV process $\{\mu(t):t\geq0\}$ whose unique stationary distribution is $\Pi(\theta\nu_{0})$. Among various definitions of this FV process, the most intuitive is in terms of its transition probability functions. In particular Ethier and Griffiths \cite{Eth(93)} shows that $\{\mu(t):t\geq0\}$ has transition function $P(t,\mu,\ddr\nu)$ given for any $t>0$ and $\mu\in\mathcal{P}_{\mathcal{X}}$ by 
\begin{equation}\label{eq_trans}
P(t,\mu,\ddr\nu)=\sum_{n\geq0}d_{n}(t)\int_{\mathcal{X}^{n}}\mu(\ddr Z_{1})\cdots\mu(\ddr Z_{n})\Pi\left(\theta\nu_{0}+\sum_{i=1}^{n}\delta_{Z_{i}}\right)(\ddr\nu).
\end{equation}
For each $t>0$ and $\mu\in\mathcal{P}_{\mathcal{X}}$, the transition probability function \eqref{eq_trans} is a (compound) mixture of distributions of Dirichlet processes. More precisely, recalling the  conjugacy property of the Dirichlet process with respect to multinomial sampling (see, e.g., Ferguson \cite{Fer(73)}), Equation \eqref{eq_trans} reads as the posterior law of a Dirichlet process with base measure $\theta\nu_{0}$, where: i) the conditioning sample is randomized with respect to the $n$-fold product measure $\mu^{n}$; ii) the sample size $n$ is randomized with respect to the marginal distribution of the ancestral process, i.e. $d_{n}(t)$ in \eqref{eq_lin_pop}.

Consider a population whose composition evolves in time according to the transition probability function \eqref{eq_trans}. Given a sample $\mathbf{Y}_{m}(t)=(Y_{1}(t),\ldots,Y_{m}(t))$ from such a population at time $t>0$, the $m$-coalescent describes the genealogy of such a sample. We denote by $C_{\mathbf{Y}_{m}(t)}$ the $m$-coalescent of the sample $\mathbf{Y}_{m}(t)$, and by $D_{m}(t)$ the number of non-mutant lineages surviving from time $0$ to time $t$ in $C_{\mathbf{Y}_{m}(t)}$. The distribution of $D_{m}(t)$ was first introduced by Griffiths \cite{Gri(80)}, and further investigated in Griffiths \cite{Gri(84)} and Tavar\'e \cite{Tav(84)}. In particular, Griffiths \cite{Gri(80)} showed that 
\begin{equation}\label{eq_lin_sam}
\mathbb{P}[D_{m}(t)=x]=(-1)^{x}\sum_{i=x}^{m}\rho_{i}(t)\frac{{m\choose i}{i\choose x}(x+\theta)_{(i-1)}}{(\theta+m)_{(i)}}
\end{equation}
for any $x=0,\ldots,m$ and each time $t>0$. If $T_{r}$ denotes the time until there are $r\geq 1$ non-mutant lineages left in the sample, then the following identity is immediate:
\begin{equation}\label{eq_time_sam}
\mathbb{P}[T_{r}\leq t]=\mathbb{P}[D_{m}(t)\leq r].
\end{equation}
Note that \eqref{eq_time_sam} with $r=1$ gives the distribution of the time of the most recent common ancestor in the sample, which is of special interest in genetic applications. For any $m\geq1$ the stochastic process $\{D_{m}(t) : t\geq0\}$ may be thought as the sampling version of the ancestral process $\{D(t) : t\geq0\}$. Indeed it can be easily verified that \eqref{eq_lin_sam} with $m=+\infty$ coincides with the probability \eqref{eq_lin_pop}. There are two different, but equivalent, ways to describe the evolution of $\{D_{m}(t) : t\geq0\}$ with respect to the transition probability functions \eqref{eq_trans}. Let $\mathbf{Y}_{m}$ be a sample from the population at time $0$, i.e. a sample from a non-atomic probability measure. The first way follows Kingman's coalescent and looks backward in time: based on  $\mathbf{Y}_{m}$, for any $t>0$ define $D^{\ast}_m(t)$ to be the total number of equivalent classes in the $m$-coalescent starting with $\{1\}, \ldots,\{m\}$, that is $D^{\ast}_m(t)$ is the total number of non-mutant ancestors of $\mathbf{Y}_{m}$ at time $-t $ in the past. An alternative way is forward looking in time and follow the lines of descent: based on  $\mathbf{Y}_{m}$, for any $t>0$ define $D^{\ast\ast}_m(t)$ to be the number of individuals that have non-mutant descendants at time $t$. It is known from the works of Griffiths \cite{Gri(80)} and Tavar\'e \cite{Tav(84)} that $D^{\ast}_m(t)$ and $D^{\ast\ast}_m(t)$ have the same distribution, which coincides with \eqref{eq_lin_sam}.

\subsection{New results on ancestral distributions}\label{subsec1}

We start by introducing a useful distributional identity for $D_{m}(t)$. For any $n\geq0$ let $(Z^{\ast}_{1},\ldots,Z^{\ast}_{n})$ be independent random variables identically distributed according to a non-atomic probability measure and, for any $m\geq1$, let $\mathbf{X}_{m}=(X_{1},\ldots,X_{m})$ be a random sample from a Dirichlet process with atomic base measure $\theta\nu_{0}+\sum_{1\leq i\leq n}\delta_{Z^{\ast}_{i}}$. The random variables $(Z^{\ast}_{1},\ldots,Z^{\ast}_{n})$ will be used to denote the genetic types of the ancestors. In order to keep track of the different ancestors, we add the non-atomic requirement for the law so that different ancestors will be represented by different types. Due to the almost sure discreteness of the Dirichlet process, the composition of the sample $\mathbf{X}_{m}$ can be described as follows. We denote by $\{X_{1}^{\ast},\ldots,X_{K_{m}}^{\ast}\}$ the labels identifying the $K_{m}$ distinct types in $\mathbf{X}_{m}$ which do not coincide with any of the atoms $Z^{\ast}_{i}$'s. Moreover, we set
\begin{itemize}
\item[i)] $\mathbf{M}_{m}=(M_{1,m},\ldots,M_{n,m})$ where $M_{j,m}=\sum_{1\leq i\leq m}\mathbbm{1}_{\{Z^{\ast}_{j}\}}(X_{i})$ denotes the number of $X_{i}$'s that coincide with the atom $Z^{\ast}_{j}$, for any $j=1,\ldots,n$;
\item[ii)] $\mathbf{N}_{m}=(N_{1,m},\ldots,N_{K_{m},m})$ where $N_{j,m}=\sum_{1\leq i\leq m}\mathbbm{1}_{\{X^{\ast}_{j}\}}(X_{i})$ denotes the number of $X_{i}$'s that coincide with the label $X^{\ast}_{j}$, for any $j=1,\ldots,K_{m}$;
\item[iii)] $V_{m}=\sum_{1\leq i\leq K_{m}}N_{i,m}$ denotes the number of $X_{i}$'s which do not coincide with any of the labels $\{Z_{1}^{\ast},\ldots,Z_{n}^{\ast}\}$.
\end{itemize}
Observe that the statistic $(\mathbf{N}_{m},\mathbf{M}_{m},K_{m},V_{m})$ includes all the information of $\mathbf{X}_{m}$, i.e., $(\mathbf{N}_{m},\mathbf{M}_{m},K_{m},V_{m})$ is sufficient for $\mathbf{X}_{m}$. See  Appendix \ref{appendix1} for a detailed description of the distribution of  $(\mathbf{N}_{m},\mathbf{M}_{m},K_{m},V_{m})$. Now, consider the random variable 
\begin{equation}\label{samp_dir_1}
R_{n,m}=\sum_{i=1}^{n}\mathbbm{1}_{\{M_{i,m}>0\}}.
\end{equation}
Precisely, $R_{n,m}$ denotes the number of distinct types in the sample $\mathbf{X}_{m}$ that coincide with the atoms $Z^{\ast}_{i}$'s. In the next theorem we derive the distribution of  $R_{n,m}$, and we introduce a distributional identity between $D_{m}(t)$ and a suitable randomization of $R_{n,m}$ with respect to $\{D(t): t\geq0\}$. See Appendix \ref{appendix1} for the proof.

\begin{thm}\label{thm2}
For any $m\geq1$ let $\mathbf{X}_{m}$ be a sample from a Dirichlet process with atomic base measure $\theta\nu_{0}+\sum_{1\leq i\leq n}\delta_{Z^{\ast}_{i}}$, for $n\geq0$. Then,  for $x=0,\ldots,\min(n,m)$
\begin{equation}\label{eq_prior_r}
\mathbb{P}[R_{n,m}=x]=x!\frac{{n\choose x}{m\choose x}(\theta+x)_{(m-x)}}{(\theta+n)_{(m)}}.
\end{equation}
Furthermore,
\begin{equation}\label{eq:id_prior}
D_{m}(t)\stackrel{\text{d}}{=}R_{D(t),m}
\end{equation}
for each $t>0$, where $\{D(t): t\geq0\}$ is the death process with marginal distribution \eqref{eq_lin_pop}.
\end{thm}

The distributional identity \eqref{eq:id_prior} introduces a Bayesian nonparametric interpretation on the sampling ancestral process $\{D_{m}(t):t\geq0\}$, in the sense that it establishes an interplay between $D_{m}(t)$ and the sampling from a Dirichlet process prior with atomic base measure $\theta\nu_{0}+\sum_{1\leq i\leq n}\delta_{Z^{\ast}_{i}}$. Intuitively, the identity \eqref{eq:id_prior} can be explained by taking the view of the forward looking description of the evolution of $\{D_{m}(t) : t\geq0\}$. Starting at time $0$ with an infinite number of individuals sampled (at random) from a non-atomic probability measure, the number of non-mutant lineages that survive at time $t>0$ is described by the ancestral process $\{D(t):t\geq0\}$. Now, consider a random sample $\mathbf{Y}_{m}$ of $m$ individuals at time 0, that is a random sample from a non-atomic probability measure which allows us to distinguish individuals. Then the sampling ancestral process $\{D_{m}(t):t\geq0\}$ describes the number of non-mutant lineages surviving at time $t$, for any $t>0$. The transition probability function \eqref{eq_trans} then says that conditionally on $D(t)=n$, the genetic types of the descendants (including mutants) of the $m$ individuals at time $t$ correspond to a sample from  a Dirichlet process prior with atomic base measure $\theta\nu_{0}+\sum_{1\leq i\leq n}\delta_{Z^{\ast}_{i}}$. That the population size remains constant comes from the Wright-Fisher approximation. Given $D(t)=n$, the genetic types of the $D_m(t)$ ancestors must have types that belong to $\{Z^{\ast}_1,\ldots, Z^{\ast}_n\}$, the distinct types of the $D(t)=n$ ancestors.  The distribution of $D_m(t)$ is independent of the exact distribution of $Z^{\ast}_1,\ldots,Z^{\ast}_n$ as long as they are distinct. Thus the conditional distribution of $D_m(t)$ given $D(t)=n$ is simply the distribution of $R_{n,m}$.

We now present a novel refinement of the ancestral distribution \eqref{eq_lin_sam}. Such a refinement takes into account the lines of descent frequencies at time $t>0$ of lines beginning at individual roots at time $0$ and surviving to time $t$; these frequencies do not include new mutants. Among the $D_{m}(t)$ non-mutant lineages surviving from time $0$ to time $t$, we denote by $D_{l,m}(t)$ the number of non-mutant lineages at time $t$ in the past having frequency $l$ at time $0$, for any $l=1,\ldots,m$. We are interested in the distribution on $D_{l,m}(t)$. Let $\mathbf{X}_{m}$ be the usual random sample from a Dirichlet process with atomic base measure $\theta\nu_{0}+\sum_{1\leq i\leq n}\delta_{Z^{\ast}_{i}}$, and define
\begin{displaymath}
R_{l,n,m}=\sum_{i=1}^{n}\mathbbm{1}_{\{M_{i,m}=l\}}.
\end{displaymath}
Precisely, $R_{l,n,m}$ denotes the number of distinct types in $\mathbf{X}_{m}$ that coincide with the atoms $Z^{\ast}_{i}$'s and have frequency $l$.  From the discussion above it is clear that, given $D(t)=n$, $D_{l,m}(t)$ has the same distribution as $R_{l,n,m}$. Thus we obtain that
\begin{equation}\label{eq:id_prior_freq}
D_{l,m}(t)\stackrel{\text{d}}{=}R_{l,D(t),m},
\end{equation}
and the marginal distribution of $\{D(t): t\geq0\}$ is given by\eqref{eq_lin_pop}. Note that the random variable $D_{l,m}(t)$ represents a natural  refinement of $D_{m}(t)$ in the sense that 
\begin{displaymath}
D_{m}(t)=\sum_{l=1}^{m}D_{l,m}(t).
\end{displaymath}
We stress the fact that $\sum_{l=1}^{m}lD_{l,m}(t)$ may be different from $m$, since frequency counts do not include new mutants. Although a large amount of literature has been devoted to the study of distributional properties of $D_{m}(t)$, to the best of our knowledge $D_{l,m}(t)$ has never been investigated, and not even introduced, before. In the next theorem we derive the distribution of $R_{l,n,m}$. See Appendix \ref{appendix1} for the proof.

\begin{thm}\label{thm1}
For any $m\geq1$ let $\mathbf{X}_{m}$ be a sample from a Dirichlet process with atomic base measure $\theta\nu_{0}+\sum_{1\leq i\leq n}\delta_{Z^{\ast}_{i}}$, for $n\geq0$. Then,  for $x=0,\ldots,\min(n,\lfloor m/l\rfloor)$
\begin{align}\label{descent_freq}
\mathbb{P}[R_{l,n,m}=x]&=\frac{m!}{(\theta+n)_{(m)}}\sum_{i=x}^{\min(n,\lfloor m/l \rfloor)}(-1)^{i-x}\frac{{i\choose x}{n\choose i}(\theta+n-i)_{(m-il)}}{(m-il)!},
\end{align}
where $\min(n,\lfloor m/l\rfloor)$ denotes the minimum between $n$ and the integer part of $m/l$.
\end{thm}

According to the distributional identity \eqref{eq:id_prior_freq}, the distribution of $D_{l,m}(t)$ follows by combining the ancestral process \eqref{eq_lin_pop} with the distribution \eqref{descent_freq}, for any $l=1,\ldots,m$. As a representative example of the distribution of $D_{l,m}(t)$, we consider the case $l=1$. The distribution of $D_{1,m}(t)$ is of special interest because it corresponds to the sampling ancestral distribution of "rare" non-mutant lineages with frequency 1, i.e. non-mutant lineages composed by a unique individual. If we apply the distribution \eqref{eq_lin_pop} to randomize $n$ in \eqref{descent_freq} with $l=1$, then
\begin{align}\label{descent_freq_t}
&\mathbb{P}[D_{1,m}(t)=x]\\
&\notag\quad=\sum_{j=x}^{m}(-1)^{j-x}{j\choose x}{m\choose j}\\
&\notag\quad\quad\times\sum_{i=j}^{m}\rho_{i}(t)(-1)^{i}\frac{{m-j\choose i-j}(\theta+i-j)_{(m-i)}(1-i-j)_{(i-j)}}{(\theta+i+j-1)_{(m-i+1)}(1-\theta+m-i)_{(i-j)}}
\end{align}
for any $x=0,\ldots,m$ and each $t>0$. The study of finite and asymptotic properties of $D_{l,m}(t)$ is out of the scope of the present paper, and it is deferred to future work. In the rest of this section we introduce a Bayesian nonparametric predictive approach to ancestral inference under the prior assumption that the composition of the population evolves in time according to \eqref{eq_trans}. In particular, we consider a sample $\mathbf{Y}_{m+m^{\prime}}(t)=(Y_{1}(t),\ldots,Y_{m}(t),Y_{m+1}(t),\ldots,Y_{m+m^{\prime}}(t))$ from such a population at time $t>0$ and we make use of \eqref{eq:id_prior} and \eqref{eq_lin_sam} to determine the conditional, or posterior, distribution of $D_{m+m^{\prime}}(t)$ given $C_{\mathbf{Y}_{m}(t)}$. A natural  refinement of this conditional distribution is also obtained by means of the identity \eqref{eq:id_prior_freq}.

\subsection{Ancestral conditional distributions}\label{subsec2}

Let $\mathbf{X}_{m}$ be a sample from a Dirichlet process with atomic base measure $\theta\nu_{0}+\sum_{1\leq i\leq n}\delta_{Z^{\ast}_{i}}$ and, for any $m^{\prime}\geq0$,  let $\mathbf{X}_{m^{\prime}}=(X_{m+1},\ldots,X_{m+m^{\prime}})$ be an additional sample. More precisely $\mathbf{X}_{m^{\prime}}$ may be viewed as a sample from the conditional distribution of the Dirichlet process with base measure $\theta\nu_{0}+\sum_{1\leq i\leq n}\delta_{Z^{\ast}_{i}}$, given the initial sample $\mathbf{X}_{m}$. We refer to Appendix \ref{appendix2} for a detailed description of the composition of $\mathbf{X}_{m^{\prime}}$. We denote by $M_{j,m^{\prime}}=\sum_{1\leq i\leq m^{\prime}}\mathbbm{1}_{\{Z^{\ast}_{j}\}}(X_{m+i})$ the number of $X_{m+i}$'s that coincide with the atom $Z^{\ast}_{j}$, and we introduce the random variable
\begin{equation}\label{eq_r_enlarge}
R_{n,m+m^{\prime}}=\sum_{i=1}^{n}\mathbbm{1}_{\{M_{i,m}+M_{i,m^{\prime}}>0\}},
\end{equation}
which denotes the number of distinct types in the enlarged sample $\mathbf{X}_{m+m^{\prime}}=\{\mathbf{X}_{m}, \mathbf{X}_{m^{\prime}}\}$ that coincide with the atoms $Z^{\ast}_{i}$'s. Observe that if we set $m^{\prime}=0$ then $R_{n,m+m^{\prime}}$ reduces to $R_{n,m}$ in \eqref{samp_dir_1}. We also introduce the following random variable 
\begin{equation}\label{eq_r_enlarge_freq}
\tilde{R}_{l,n,m^{\prime}}=\sum_{i=1}^{n}\mathbbm{1}_{\{M_{i,m^{\prime}}>0\}}\mathbbm{1}_{\{M_{i,m}=l\}},
\end{equation}
which is the number of distinct types in the additional sample $\mathbf{X}_{m^{\prime}}$ that coincide with the atoms $Z^{\ast}_{i}$'s that have frequency $l$ in the samples $\mathbf{X}_{m}$. In the next theorem we derive the conditional distribution of $R_{n,m+m^{\prime}}$ given $(\mathbf{N}_{m},\mathbf{M}_{m},K_{m})$, and the conditional distribution of $\tilde{R}_{l,n,m^{\prime}}$ given $(\mathbf{N}_{m},\mathbf{M}_{m},K_{m})$. Interestingly, it turns out that such conditional distributions depend on $(\mathbf{N}_{m},\mathbf{M}_{m},K_{m})$ solely through the statistics $R_{n,m}$ and $R_{l,n,m}$, respectively. See Appendix \ref{appendix2} for the proof.

\begin{thm}\label{prp1} For any $m\geq1$ and $m^{\prime}\geq0$ let $\mathbf{X}_{m+m^{\prime}}$ be a sample from a Dirichlet process with atomic base measure $\theta\nu_{0}+\sum_{1\leq i\leq n}\delta_{Z^{\ast}_{i}}$, for $n\geq0$. Then one has
\begin{itemize}
\item[i)] for $x=y,\ldots,\min(n,y+m^{\prime})$
\begin{align}\label{eq_post_r}
&\mathbb{P}[R_{n,m+m^{\prime}}=x\,|\,\mathbf{N}_{m}=\mathbf{n}_{m},\mathbf{M}_{m}=\mathbf{m}_{m},K_{m}=k_{m}]\\
&\notag\quad=\mathbb{P}[R_{n,m+m^{\prime}}=x\,|\,R_{n,m}=y]\\
&\notag\quad=(x-y)!\frac{{n-y\choose x-y}{m^{\prime}\choose x-y}(\theta+m+x)_{(m^{\prime}-x+y)}}{(\theta+n+m)_{(m^{\prime})}};
\end{align}
\item[ii)] for $x=0,\ldots,\min(y,m^{\prime})$
\begin{align}\label{eq_post_r_freq}
&\mathbb{P}[\tilde{R}_{l,n,m^{\prime}}=x\,|\,\mathbf{N}_{m}=\mathbf{n}_{m},\mathbf{M}_{m}=\mathbf{m}_{m},K_{m}=k_{m}]\\
&\notag\quad=\mathbb{P}[\tilde{R}_{l,n,m^{\prime}}=x\,|\,R_{l,n,m}=y]\\
&\notag\quad=\frac{{y\choose x}}{(\theta+n+m)_{(m^{\prime})}}\\
&\notag\quad\quad\times\sum_{i=y-x}^{y}(-1)^{i-(y-x)}{x\choose y-i}(\theta+n+m-i(1+l))_{(m^{\prime})}.
\end{align}
\end{itemize}
Therefore, $R_{n,m}$ and $R_{l,n,m}$ are sufficient to predict $R_{n,m+m^{\prime}}$ and $\tilde{R}_{l,n,m^{\prime}}$, respectively.
\end{thm}

The predictive sufficiency of $R_{n,m}$ in \eqref{eq_post_r} plays a fundamental role for deriving the conditional counterpart of the sampling ancestral distribution \eqref{eq_lin_sam}. In particular, consider a population whose composition evolves in time according to \eqref{eq_trans}, and let $\mathbf{Y}_{m}(t)$ be a sample of $m$ individuals from the population at time $t$. Furthermore, for any $m^{\prime}>1$ let $\mathbf{Y}_{m^{\prime}}(t)=(Y_{m+1}(t),\ldots,Y_{m+m^{\prime}}(t))$ be an additional unobservable sample. The identity \eqref{eq:id_prior} and the sufficiency of $R_{n,m}$ imply that the conditional distribution of $D_{m+m^{\prime}}(t)$ given $D_{m}(t)$ can be obtained by randomizing the parameter $n$ in \eqref{eq_post_r} with respect to the distribution
\begin{equation}\label{eq_bayes}
\mathbb{P}[D(t)=n\,|\,D_{m}(t)=y]=\frac{\mathbb{P}[R_{n,m}=y]\mathbb{P}[D(t)=n]}{\mathbb{P}[D_{m}(t)=y]},
\end{equation}
where $\mathbb{P}[D_{m}(t)=y]=\sum_{n\geq0}\mathbb{P}[R_{n,m}=y]\mathbb{P}[D(t)=n]$, and the distributions of $R_{n,m}$, $D_{m}(t)$ and $D(t)$ are in \eqref{eq_prior_r}, \eqref{eq_lin_sam} and \eqref{eq_lin_pop}, respectively. Then, we can write
\begin{align}\label{eq_predictive}
&\mathbb{P}[D_{m+m^{\prime}}(t)=x\,|\, C_{\mathbf{Y}_{m}(t)}]\\
&\notag\quad=\mathbb{P}[D_{m+m^{\prime}}(t)=x\,|\,D_{m}(t)=y]\\
&\notag\quad=\frac{{m\choose y}{m^{\prime}\choose x-y}(\theta+y)_{(x-y)}(m+m^{\prime}+\theta)_{(y)}\mathbb{P}[D_{m+m^{\prime}}(t)=x]}{{m^{\prime}+m\choose x}(\theta+m)_{(x)}\mathbb{P}[D_{m}(t)=y]}
\end{align}
for any $x=0,\ldots,m+m^{\prime}$ and each $t>0$. Note that the probability \eqref{eq_predictive} with $m=0$ reduces to the unconditional ancestral distribution \eqref{eq:id_prior}. Moments of \eqref{eq_predictive} are obtained by randomizing $n$, with respect to \eqref{eq_bayes}, in the corresponding moments of  \eqref{eq_post_r} given in \eqref{rth_fac_mom}. Equation \eqref{eq_predictive} introduces a novel sampling ancestral distribution under the Kingman coalescent. Observe that, due to the identity \eqref{eq_time_sam}, the distribution \eqref{eq_predictive} leads to the conditional distribution of the time $T_{r}$ until there are $r$ non-mutant lineages left in the $(m+m^{\prime})$-coalescent.

We now consider a refinement of \eqref{eq_predictive} which takes into account the frequency counts of non-mutant lineages. Specifically, we determine the conditional distribution of the number $\tilde{D}_{l,m^{\prime}}(t)$ of non-mutant lineages surviving from time $0$ to time $t$ in $\mathbf{Y}_{m^{\prime}}(t)$ whose frequency in the lineages ancestral to the initial sample $\mathbf{Y}_{m}$ is $l$. As a representative example we focus on $l=1$. Due to \eqref{eq:id_prior_freq} and the sufficiency of $R_{l,n,m}$ to predict $\tilde{R}_{l,n,m^{\prime}}$, the conditional distribution of $\tilde{D}_{1,m^{\prime}}(t)$ given $D_{1,m}(t)$ is obtained by randomizing the parameter $n$ in \eqref{eq_post_r_freq} with respect to the distribution
\begin{equation}\label{eq_bayes_freq}
\mathbb{P}[D(t)=n\,|\,D_{1,m}(t)=y]=\frac{\mathbb{P}[R_{1,n,m}=y]\mathbb{P}[D(t)=n]}{\mathbb{P}[D_{1,m}(t)=y]},
\end{equation}
because $\mathbb{P}[D_{1,m}(t)=y]=\sum_{n\geq0}\mathbb{P}[R_{1,n,m}=y]\mathbb{P}[D(t)=n]$, and the distributions of $R_{1,n,m}$, $D_{1,m}(t)$ and $D(t)$ are in \eqref{descent_freq}, \eqref{descent_freq_t} and \eqref{eq_lin_pop}, respectively. Then, we have
\begin{align}\label{eq_predictive_2}
&\mathbb{P}[\tilde{D}_{1,m^{\prime}}(t)=x\,|\,C_{\mathbf{Y}_{m}(t)}]\\
&\notag\quad=\mathbb{P}[\tilde{D}_{1,m^{\prime}}(t)=x\,|\,D_{1,m}(t)=y]\\
&\notag\quad=(-1)^{x}\frac{{y\choose x}}{\mathbb{P}[D_{1,m}(t)=y]}\sum_{k=0}^{y}(-1)^{y-k}{x\choose y-k}\\
&\notag\quad\quad\times\sum_{j=y}^{m}(-1)^{j-y}{j\choose y}{m\choose j}\\
&\notag\quad\quad\quad\times\sum_{i=j}^{m}\frac{\rho_{i}(t)}{(i-j)!}\sum_{n=j}^{i}(-1)^{n}\frac{{i-j\choose i-n}(\theta+n-j)_{(m-j)}(\theta+n+m-2k)_{(m^{\prime})}}{(\theta+n+i-1)_{(m+m^{\prime}-i+1)}},
\end{align}
for any $x=0,\ldots,m^{\prime}$ and each $t>0$. Moments of \eqref{eq_predictive_2} are obtained by randomizing the parameter $n$, with respect to \eqref{eq_bayes_freq}, in the corresponding moments of  \eqref{eq_post_r_freq} given in \eqref{rth_fac_mom_freq}.  We stress that the sufficiency of $R_{l,n,m}$ to predict $\tilde{R}_{l,n,m^{\prime}}$ plays a fundamental role for determining the conditional distribution of $\tilde{D}_{l,m^{\prime}}(t)$.

If we interpret the FV transition probability function \eqref{eq_trans} as a prior distribution on the evolution in time of the composition of the population, then the conditional distributions \eqref{eq_predictive} and  \eqref{eq_predictive_2} take on a natural Bayesian nonparametric meaning. Specifically, they correspond to the posterior distributions of $D_{m+m^{\prime}}(t)$ and $\tilde{D}_{1,m}(t)$, respectively, given the initial sample $\mathbf{Y}_{m}(t)$ whose ancestry $C_{\mathbf{Y}_{m}(t)}$ features $D_{m}(t)$ non-mutant lineages of which $D_{1,m}(t)$ are of frequency $1$. Given the information on $D_{m}(t)$ and $D_{1,m}(t)$ from the initial observed sample, the expected values of  \eqref{eq_predictive} and \eqref{eq_predictive_2} provide us with Bayesian nonparametric estimators, under a squared loss function, of $D_{m+m^{\prime}}(t)$ and $\tilde{D}_{1,m}(t)$.  It is worth pointing out that $D_{m}(t)$ and $D_{1,m}(t)$, and in general the $m$-coalescent $\{C_{\mathbf{Y}_{m}(t)}:t\geq0\}$, are latent quantities, in the sense that they are not directly observable from the data. However, one can easily infer $D_{m}(t)$ and $D_{1,m}(t)$, as well as the mutation parameter $\theta$, from the observed data and then combine their estimates with the posterior distributions \eqref{eq_predictive} and \eqref{eq_predictive_2}. This approach for making predictive ancestral inference will be detailed in Section \ref{sec3}. We conclude this section with a proposition that introduces an interesting special case of the posterior distributions \eqref{eq_predictive} and  \eqref{eq_predictive_2}. See Appendix \ref{appendix2} for the proof. Let
\begin{displaymath}
\tilde{D}_{m^{\prime}}(t)=D_{m+m^{\prime}}(t)-D_{m}(t)
\end{displaymath}
be the number of new non mutant lineages, that is $\tilde{D}_{m^{\prime}}(t)$ is the number non-mutant lineages at $t$ back in the additional sample of size $m^{\prime}$ which do not coincide with any of the non-mutant lineages at time $t$ back in the initial sample of size $m$.

\begin{prp}\label{cor1} Consider a population whose composition evolves in time according to the FV transition probability function \eqref{eq_trans}. Then for each $t>0$ one has
\begin{align}\label{gt_cor_1}
&\mathbb{P}[\tilde{D}_{1}(t)=1\,|\,C_{\mathbf{Y}_{m}(t)}]\\
&\notag\quad=\mathbb{P}[\tilde{D}_{1}(t)=1\,|\,D_{m}(t)=y]\\
&\notag\quad=\frac{(y+1)(\theta+y)\mathbb{P}[D_{m+1}(t)=y+1]}{(m+1)(\theta+m)\mathbb{P}[D_{m}(t)=y]}
\end{align}
and
\begin{align}\label{gt_cor_2}
&\mathbb{P}[\tilde{D}_{1,1}(t)=1\,|\,C_{\mathbf{Y}_{m}(t)}]\\
&\notag\quad=\mathbb{P}[\tilde{D}_{1,1}(t)=1\,|\,D_{1,m}(t)=y]\\
&\notag\quad=\frac{y}{\mathbb{P}[D_{1,m}(t)=y]}\\
&\notag\quad\quad\times\sum_{j=y}^{m}(-1)^{j-y}{j\choose y}{m\choose j}\\
&\notag\quad\quad\quad\times\sum_{i=j}^{m}\frac{\rho_{i}(t)}{(i-j)!}\sum_{n=j}^{i}(-1)^{n}\frac{{i-j\choose i-n}(\theta+n-j)_{(m-j)}}{(\theta+n+i-1)_{(m+1-i+1)}}.
\end{align}
\end{prp}

Proposition \ref{cor1} introduces two Bayesian nonparametric estimators for the probability of discovering non-mutant lineages surviving from time $0$ to time $t>0$. This proposition makes explicit the link between our results and the work of Good \cite{Goo(53)}, where the celebrated Good-Turing estimator has been introduced. Given a sample of size $m$ from a population of individuals belonging to an (ideally) infinite number of species with unknown proportions, the Good-Turing estimator provides with an estimate of the probability of discovering at the $(m+1)$-th draw a species observed with frequency $l$ in the initial sample. Of course $l=0$ corresponds to the case of the probability of discovering a new species at the $(m+1)$-th draw. Within our framework for ancestral inference under the FV prior assumption \eqref{eq_trans}, the probabilities \eqref{gt_cor_1} and \eqref{gt_cor_2} may be considered as natural Bayesian nonparametric counterparts of the celebrated Good-Turing estimators. Precisely: \eqref{gt_cor_1} is the Bayesian nonparametric estimator of the probability of discovery in one additional sample a new non-mutant lineage surviving from time $0$ to time $t>0$; \eqref{gt_cor_2} is the Bayesian nonparametric estimator of the probability of discovery in one additional sample a non-mutant lineages surviving from time $0$ to time $t>0$ and whose frequency is $1$ in the initial sample.


\section{Illustration}\label{sec3}
In this section we show how to use the results of the previous section by applying them to real genetic dataset. Consider a population whose composition evolves in time according to the FV transition probability function  \eqref{eq_trans},  and suppose we observe a sample of $m$ individuals $\mathbf{Y}_{m}$ taken from a Dirichlet process with base measure $\theta\nu_{0}$. Recall that, under this assumption on the evolution of the population, the law of the Dirichlet process with base measure $\theta\nu_{0}$ is the unique stationary distribution of the neutral FV process. The sample then consists of a collection of $K_{m}=k\leq m$ distinct genetic types with corresponding frequencies $(N_{1},\ldots,N_{k})=(n_{1},\ldots,n_{k})$. In particular if $p^{(m)}(n_{1},\ldots,n_{k})$ denotes the probability of a sample $\mathbf{Y}_{m}$, which features $k$ genetic types with frequencies $(n_{1},\ldots,n_{k})$, then 
\begin{equation}\label{eq_like}
p^{(m)}(n_{1},\ldots,n_{k})=\frac{\theta^{k}}{(\theta)_{m}}\prod_{i=1}^{k}(n_{i}-1)!;
\end{equation}
see Ewens \cite{ewe:1972} for details. With a slight abuse of notation, we denote by $X\,|\,Y$ a random variable whose distribution coincides with the conditional distribution of $X$ given $Y$. As we pointed out at the end of Section \ref{sec2}, in order to apply the posterior distributions \eqref{eq_predictive} and \eqref{eq_predictive_2} we have to estimate the unobservable quantities $(\theta,D_{m}(t))$ and $(\theta,D_{1,m}(t))$, respectively. Using a fully Bayesian approach, estimates of $(\theta,D_{m}(t))$ and $(\theta,D_{1,m}(t))$ are obtained as the expected values of the posterior distributions of $(\theta,D_{m}(t))$ and $(\theta,D_{1,m}(t))$ given $\mathbf{Y}_{m}$, with respect to some prior choice for $\theta$. For simplicity we focus on the posterior distributions of $D_{m}(t)\,|\,\mathbf{Y}_{m}$ and $D_{1,m}(t)\,|\,\mathbf{Y}_{m}$, and we resort to an empirical Bayes approach for estimating $\theta$.  Specifically, we use the maximum likelihood estimate for $\theta$ originally proposed by Ewens \cite{ewe:1972}, which is obtained from the likelihood \eqref{eq_like} by numerically finding the root of a certain polynomial in $\theta$. From the point of view of ancestral inference, the distributions of  $D_{m}(t)\,|\,\mathbf{Y}_{m}$ and $D_{l,m}(t)\,|\,\mathbf{Y}_{m}$ correspond respectively to the questions: How many non-mutant genetic ancestors to the sample existed a time $t$ ago? And how many non-mutant genetic ancestors existed whose type appeared with frequency $l$ among those ancestors?

First we consider the posterior distribution of $D_m(t)\,|\,\mathbf{Y}_{m}$. Under the Kingman coalescent model in which mutation is parent-independent, the distribution of this random variable is straightforward: indeed it is well known that the distribution of $D_m(t)\,|\,\mathbf{Y}_{m}$ coincides with the distribution of $D_{m}(t)$, for any $t>0$. This holds because the coalescent process for a sample of size $m$ can be decomposed into its ancestral process $\{D_m(t): t \geq 0\}$, and a skeleton chain taking values in marked partitions of the set $\{1,\dots,m\}$. See Watterson \cite{wat:1984} and references therein for details. These two processes are independent, and the sample $\mathbf{Y}_m$ is informative about only the skeleton chain. Thus, the distribution of $D_m(t)\,|\,\mathbf{Y}_{m}$ is given by \eqref{eq_lin_sam}. In particular if we denote by $\hat{\theta}$ the maximum likelihood estimate of $\theta$, then an estimate of $D_{m}(t)$ is given by the following expression,
\begin{displaymath}
\hat{D}_{m}(t)=\sum_{i=1}^{m}\rho_{i}(t)(-1)^{i}i!\frac{{m\choose i}}{(\hat{\theta}+m)_{(i)}},
\end{displaymath}
which is the expected value of the distribution \eqref{eq_lin_sam} with $\theta$ replaced by its estimate $\hat{\theta}$. Thus, we can plug in the estimate $(\hat{\theta},\hat{D}_{m}(t))$ to the posterior distribution \eqref{eq_predictive} and then predict the number of non-mutant lineages surviving from time $0$ to time $t$ in the enlarged sample of size $m+m^{\prime}$, given the initial observed sample of size $m$. Observe that under parent-independent mutation the information of the initial sample $\mathbf{Y}_{m}$ affects the prediction only through the estimate $\hat{\theta}$. 

The posterior distribution of $D_{l,m}(t)\,|\,\mathbf{Y}_{m}$ is not trivial. Differently from the distribution of $D_m(t)\,|\,\mathbf{Y}_{m}$, which is independent of $\mathbf{Y}_{m}$, the sample $\mathbf{Y}_{m}$ is informative for $D_{l,m}(t)$. In order to derive the distribution of $D_{l,m}(t)\,|\,\mathbf{Y}_{m}$, one strategy would be to study a posterior analogue of the marked-partition-valued process introduced in the work by Watterson \cite{wat:1984}, and then project it onto its block sizes. The resulting formulas are, however, unwieldy. Our preferred approach is via Monte Carlo simulation, since the posterior transition rates of $D_{l,m}(t)\,|\,\mathbf{Y}_{m}$ evolving backwards in time are easy to describe. In particular, if we set 
\begin{displaymath}
\bfD_{\cdot,m} = \left\{\left(D_{1,m},D_{2,m},\dots,D_{m,m}\right)(t): t \geq 0\right\}
\end{displaymath}
then the transition rate matrix for $\bfD_{\cdot,m}$ when currently $\sum_{1\leq l\leq m}lD_{l,m}(t) =x$ is
\begin{align*}
&q_{\bfD_{\cdot,m},\bfD_{\cdot,m}'}\\ 
&\notag=\frac{x(x+\theta-1)}{2}\\
&\notag\quad\times\begin{cases}
\dfrac{lD_{l,m}(t)}{x} & \text{if }\bfD_{\cdot,m}'(t) = (D_{1,m},\dots, D_{l-1,m} + 1,D_{l,m} - 1,\dots,D_{m,m})(t)\\
&\phantom{\text{if }}l=2,\dots,m\\
\dfrac{D_{1,m}(t)}{x} & \text{if }\bfD_{\cdot,m}'(t) = (D_{1,m}-1,\dots,D_{m,m})(t)\\
-1 & \text{if }\bfD_{\cdot,m}'(t) = \bfD_{\cdot,m}(t)\\
0 & \text{otherwise},
\end{cases}
\end{align*}
with initial condition
\begin{displaymath}
(D_{l,m}(0)\,|\,\mathbf{Y}_{m}) = \left|\left\{y \in \mathcal{X}: \sum_{i=1}^m \mathbbm{1}_{\{y\}}(Y_i) = l\right\}\right|.
\end{displaymath}
See Hoppe \cite{hop:1987} for details on $q_{\bfD_{\cdot,m},\bfD_{\cdot,m}'}$. In words, lineages are lost at rate $x(x+\theta-1)/2$. At such an event, the lineage selected to be lost is chosen uniformly at random. If that lineage contributed to $D_{1,m}(t)$ then we recognise this loss as having been caused by mutation, otherwise its loss was due to coalescence. The stochastic process $\bfD_{\cdot,m}$ can be regarded as a time-evolving counterpart to the allelic random partition introduced by Ewens \cite{ewe:1972}, whose stationary sampling distribution is the Ewens-sampling formula, and is clearly straightforward to simulate.

We now present a numerical illustration of our approach. We reconsider the electrophoretic dataset in Table 1 of Singh et al.\ \cite{sin:etal:1976}, who sampled $m = 146$ family lines of the fruit fly \emph{Drosophila pseudoobscura} at the xanthine dehydrogenase locus. This organism is well studied in evolutionary biology, and it is especially used to address questions on the nature of speciation. It is thought to have diverged from its sister species \emph{Drosophila persimilis} about 589,000 years ago. See Hey and Nielsen \cite{hey:nie:2004} and references therein for details. It is therefore important to quantify relative levels of genetic diversity either shared between the two species or private to one of them. One might ask how much of the genetic diversity observed by Singh et al.\ \cite{sin:etal:1976} existed intact at the time $t_\text{div}$ the two species diverged. In other words, what is the distribution of $(\bfD_{\cdot,m}(t_\text{div})\,|\, \mathbf{Y}_m)$? The process commences back in time from an initial allelic partition inferred from the sample $\mathbf{Y}_m$:
\begin{displaymath}
\bfD_{\cdot,146}(0) = (10,3,7,0,2,2,0,1,0,0,1,0,\dots, 0,1,0,\dots, 0),
\end{displaymath}
where the most common allele (the rightmost 1) has multiplicity equal to 68. The reader is referred to the work of Singh et al. \cite{sin:etal:1976} for details on these data. Using a different dataset, Hey and Nielsen \cite{hey:nie:2004} estimated $t_\text{div} = 0.34$ in units of $2N_e$ generations, where $N_e$ is the diploid effective population size (these units are appropriate when appealing to the coalescent timescale). The use of the  maximum likelihood estimate $\hat{\theta} = 9.5$, and the application of the simulation process described above, result in a Monte Carlo sample of $\bfD_{\cdot,m}(t_\text{div})\,|\,\mathbf{Y}_m$ which is summarized in Figure \ref{fig:Dhists}. In particular,  from Figure \ref{fig:Dhists}, posterior means are $\hat{D}_m(t_\text{div}) = 2.31$, with narrowest 95\% credible interval $[0,4]$; and $\hat{D}_{1,m}(t_\text{div}) = 1.58$, with a narrowest 95\% credible interval $[0,3]$. In other words, with high probability almost all genetic variability, as summarised by the total number of lineages $D_m(t)$ and the total number of singleton lineages $D_{1,m}(t)$, is lost as far back as $t_\text{div}$.

\medskip
\begin{figure}[htbp]
\begin{center}
\begin{tabular}{cc}
(a) & (b)\\[5pt]
\includegraphics[width=0.45\textwidth]{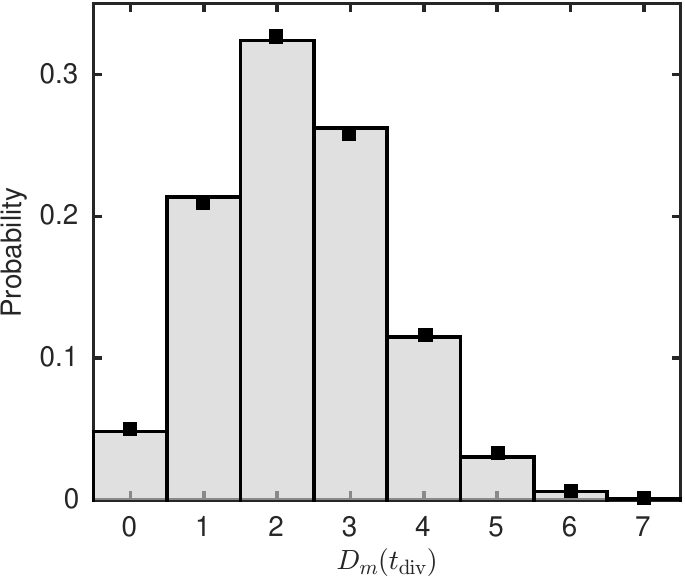} & \includegraphics[width=0.45\textwidth]{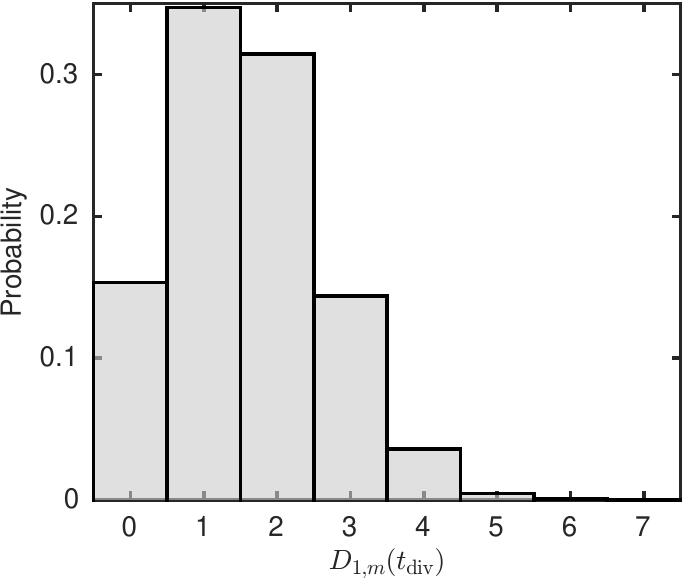}
\end{tabular}
\caption{\label{fig:Dhists}An approximation of $\bfD_{\cdot,m}(t_\text{div})\,|\,\mathbf{Y}_m$ using the data of Singh et al.\ \cite{sin:etal:1976} and  $10^4$ Monte Carlo replicates, summarized by (a) $D_m(t_\text{div})\,|\,\mathbf{Y}_m$ and (b) $D_{1,m}(t_\text{div})\,|\,\mathbf{Y}_m$. Also shown are the predictions from equation \eqref{eq_lin_sam} (black squares).}
\end{center}
\end{figure}

As discussed above, Equation \eqref{eq_predictive} and Equation \eqref{eq_predictive_2} provide us with a quick predictive distribution for the following question: if we take an additional sample of size $m^{\prime}$, how much additional genetic variability in the historical population that existed at the divergence time is uncovered? This question is informative because it provides a window into levels of diversity in an unobservable historical population with respect to alleles existing in the modern day. This in turn governs the levels of divergence that we might expect between the two modern species. Equation \eqref{eq_predictive} provides a distribution on the total number of non-mutant lineages ancestral to the enlarged sample given $D_m(t)$ lineages ancestral to the original sample, while Equation \eqref{eq_predictive_2} provides a distribution on the number of singleton (frequency $1$) lineages ancestral to the original sample that are also discovered in the additional sample. If we plug the (rounded) posterior means $\hat{D}_m(t_\text{div}) = 2$ and $\hat{D}_{1,m}(t_\text{div}) = 2$ to \eqref{eq_predictive} and \eqref{eq_predictive_2} respectively, along with the maximum likelihood estimate $\hat{\theta} = 9.48$, then we obtain the predictive distributions shown in Figure \ref{fig:mprime}. It is clear that, if we regard increasing the initial sample size by $m^{\prime}$ as a method of ``ancestral lineage discovery'', then this method is rather inefficient. With high probability, the total number of ancestral lineages is still two, and at most increases to three, even if the sample size is increased by 50. Figure \ref{fig:mprime}(b) shows that at least some of this inefficiency is due to the fact that the two singleton alleles ancestral to the original sample are also ancestral to members of the additional sample; it is moderately easy for these alleles to be rediscovered in the additional sample, at least for sufficiently large $m'$. Note that these observations are not surprising since the additional lineages coalesce rapidly with each other as we go back in time. Because of shared ancestry, taking additional samples in a coalescent framework is far less informative than the random sampling typically possible in other statistical models.

\medskip
\begin{figure}[htbp]
\begin{center}
\begin{tabular}{cc}
(a) & (b)\\[5pt]
\includegraphics[width=0.45\textwidth]{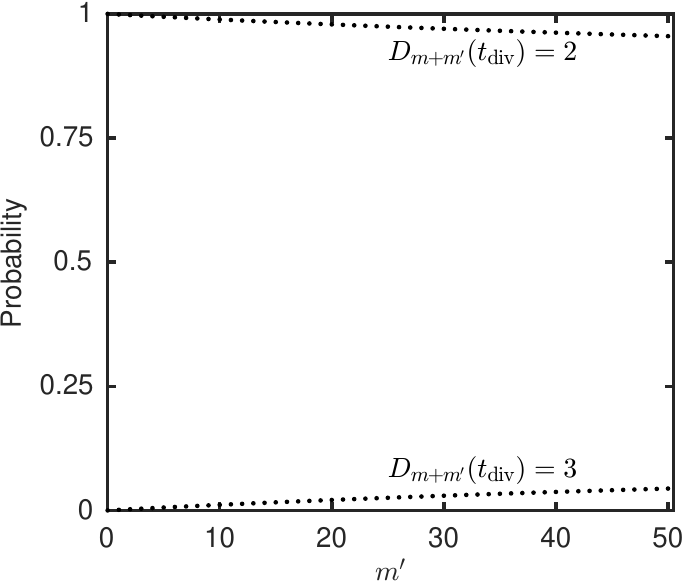} & \includegraphics[width=0.45\textwidth]{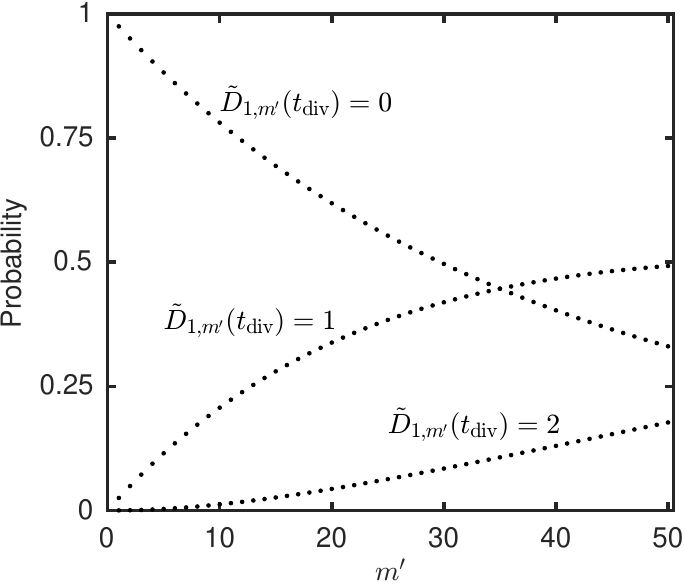}
\end{tabular}
\caption{\label{fig:mprime}(a) The probability $\bbP[D_{m+m^{\prime}}(t_\text{div})=x\mid D_m(t_\text{div}) = 2]$ that there are $x$ non-mutant lineages ancestral to an enlarged sample of size $m+m^{\prime}$, given that there were two lineages ancestral to the original sample. (b) The probability $\mathbb{P}[\tilde{D}_{1,m^{\prime}}(t_\text{div})=x\,|\,D_{1,m}(t_\text{div})=2]$ that $x$ of the two singleton lineages ancestral to an original sample of size $m$ are also ancestral to some members of the additional sample of size $m^{\prime}$.}
\end{center}
\end{figure}

It is worth remarking that the idea of estimating the allelic configuration of an unobservable historical population, as described above, has broader utility. Very sophisticated models have been formulated in population genetics, encompassing a variety of phenomena we have ignored in this paper: nucleotide-level mutation, changes in historical population size, population substructure, and so on. It turns out that inference under these models can be phrased in terms of the predictive, or conditional sampling, distributions associated with an additional sample of size $m^{\prime}$, which in turn depend on the genetic types of lineages ancestral to the original sample. See Stephens and Donnelly \cite{Ste(00)} for a detailed account. Under these more sophisticated models, such predictive distributions are usually intractable, but Stephens and Donnelly \cite{Ste(00)} showed that if they can be approximated then exact Monte Carlo based inference is still possible by applying an importance sampling correction to this approximation. Thus, many population genetic inference problems can be reduced to the following: design a decent approximation to the predictive distribution associated with an enlarged sample. See, e.g., De Iorio and Griffiths \cite{dei:gri:2004:I}, Griffiths et al. \cite{gri:etal:2008}, Hobolth et al. \cite{hob:etal:2008}, Li and Stephens \cite{li:ste:2003}, Paul and Song \cite{pau:son:2010}, Paul et al. \cite{pau:etal:2011}, Sheehan et al. \cite{she:etal:2013} and Stephens and Donnelly \cite{Ste(00)}. Now the tractability of the model studied in this paper becomes crucial: it can be used as a guide for more complex models. Indeed, this is the strategy taken by Stephens and Donnelly \cite{Ste(00)} and Hobolth et al. \cite{hob:etal:2008}.


\section{Discussion}\label{sec4}

We introduced a Bayesian nonparametric predictive approach to ancestral inference. This approach relies on the FV transition probability function \eqref{eq_trans} as a nonparametric prior assumption for the evolution in time of the composition of a genetic population. That is, backward in time the Kingman coalescent is assumed to be the prior model for the genealogy of the population. Under this prior assumption, and given a sample of $m$ individuals from the population at time $t>0$, we showed how to derive the posterior distributions of some quantities related to the genealogy of an additional unobservable sample of size $m^{\prime}\geq1$. Our posterior analysis built upon the distributional identity for $D_{m}(t)$ introduced in Theorem \ref{thm1}, which provides a Bayesian nonparametric interpretation of the sampling ancestral process. In particular, we determined the posterior distribution of the number $D_{m+m^{\prime}}(t)$ of non-mutant lineages surviving from time $0$ to time $t$ in the enlarged sample of size $m+m^{'}$ which, in turn, leads to the posterior distribution of the time of the most recent common ancestor. This result has then been extended to the number $\tilde{D}_{l,m^{\prime}}(t)$ of non-mutant lineages having frequency $l$, for any $l=1,\ldots,m$, surviving from time $0$ to time $t$. Our results allowed us to introduce a novel class of Bayesian nonparametric estimators which can be thought as Good-Turing estimators in the context of ancestral inference.

This paper paves the way for future work towards predictive ancestral inference under the Kingman coalescent. A first important problem consists in investigating the asymptotic behavior of the statistics introduced in this paper. The asymptotic behaviour of $D_{m}(t)$ for small time $t$ was first investigated by Griffiths \cite{Gri(84)}. If $m\rightarrow+\infty$ and $t\rightarrow0$ such that $mt$ is constant, then $D_{m}(t)$ appropriately scaled converges in distribution to a Gaussian random variable. Furthermore, if $t$ goes to $0$ faster with $m^{2}t$ being bounded above, then $m-D_{m}(t)$ will be approximated in distribution by a Poisson random variable. Besides extending these asymptotic results to $D_{l,m}(t)$, it would be interesting to characterize the large $m^{\prime}$ asymptotic behavior of the posterior distributions of $D_{m+m^{\prime}}(t)$ and $\tilde{D}_{l,m^{\prime}}(t)$. Such a characterization would be useful to obtain large $m^{\prime}$ approximations of these posterior distributions. Work on this is ongoing. Another problem consists in investigating the posterior distribution of other statistics of the additional sample. Apart from statistics related to the age of alleles, it seems natural to complete our analysis by determining a posterior counterpart of the distribution of $D_{l,m}(t)$. This requires one to study the conditional distribution of $R_{l,n,m+m^{\prime}}=\sum_{1\leq i\leq n}\mathbbm{1}_{\{M_{i,m}+M_{i,m^{\prime}}=l\}}$ given $\mathbf{X}_{m}$. While this conditional distribution can be derived by means of techniques similar to those developed in this paper, we expect that $R_{l,n,m+m^{\prime}}$ will be a function of $\mathbf{X}_{m}$ through $R_{n,m}$ and $(R_{1,n,m},\ldots,R_{l,n,m})$. Such unwieldy sufficient statistics could make difficult the randomization of $R_{l,n,m+m^{\prime}}\,|\,R_{n,m},(R_{1,n,m},\ldots,R_{l,n,m})$ over the parameter $n$.

Kingman's coalescent is a special case of a broader class of so-called $\Lambda$-coalescent models, which have been generalized further to the $\Xi$-coalescents, whose genealogies allow for simultaneous coalescence events each involving possibly more than two lineages. We refer to the monograph by Berestycki \cite{Ber(09)} for a comprehensive and stimulating account of these generalizations of the Kingman coalescent. In particular $\Xi$-coalescents arise in genetics as models of diploid populations with high fecundity and highly skewed offspring distributions; that is, in some generations the offspring of a single individual can replace a substantial fraction of the whole population. See M\"ohle and Sagitov \cite{moe:sag:2003} for a biological interpretation of the $\Xi$-coalescent. Another direction for future work is in Bayesian nonparametric predictive ancestral inference for these broader classes of coalescent models. This seems to be a much more challenging task, as far fewer results are available. For example, M\"ohle \cite{moe:2006} showed that there is no simple analogue of the Ewens-sampling formula except in a few special cases. Furthermore, although there exists a construction of the $\Xi$-Fleming-Viot process to describe the forwards-in-time evolution of the population, there seems to be no tractable expressions for its transition function as in \eqref{eq_trans}, nor even a known stationary distribution in general. See Birkner et al. \cite{bir:etal:2009} for details. It may be therefore difficult to give a natural Bayesian nonparametric interpretation of the underlying genealogical process.


\appendix
\section*{Appendix}
\numberwithin{equation}{section}
\numberwithin{thm}{section}
\numberwithin{lem}{section}
\numberwithin{prp}{section}

\subsection{Proofs of Section \ref{subsec1}}\label{appendix1}

Let $\mathbf{X}_{m}$ be a sample from a Dirichlet process with atomic base measure $\theta\nu_{0}+\sum_{1\leq i\leq n}\delta_{Z^{\ast}_{i}}$. Recall from Section \ref{subsec1} that we $\{X_{1}^{\ast},\ldots,X_{K_{m}}^{\ast}\}$ the labels identifying the $K_{m}$ distinct types in $\mathbf{X}_{m}$ which do not coincide with any of the atoms $Z^{\ast}_{i}$'s. Moreover, we defined the following quantities: i) $\mathbf{M}_{m}=(M_{1,m},\ldots,M_{n,m})$ where $M_{j,m}=\sum_{1\leq i\leq m}\mathbbm{1}_{\{Z^{\ast}_{j}\}}(X_{i})$ denotes the number of $X_{i}$'s that coincide with the atom $Z^{\ast}_{j}$, for any $j=1,\ldots,n$; ii) $\mathbf{N}_{m}=(N_{1,m},\ldots,N_{K_{m},m})$ where $N_{j,m}=\sum_{1\leq i\leq m}\mathbbm{1}_{\{X^{\ast}_{j}\}}(X_{i})$ denotes the number of $X_{i}$'s that coincide with the label $X^{\ast}_{j}$, for any $j=1,\ldots,K_{m}$; iii) $V_{m}$ denotes the number of $X_{i}$'s which do not coincide with any of the labels $\{Z_{1}^{\ast},\ldots,Z_{n}^{\ast}\}$, i.e., $V_{m}=\sum_{1\leq i\leq K_{m}}N_{i,m}$. Note that the sample $\mathbf{X}_{m}$ may be viewed as a a random sample from a posterior Dirichlet process given the sample $(Z_{1}^{\ast},\ldots,Z^{\ast})$ featuring $n$ distinct types. Since the distribution of a sample of size $n$ from a Dirichlet process (Ewens \cite{ewe:1972}) featuring $J_{n}$ distinct types with corresponding frequency $\mathbf{Q}_{n}=(Q_{1,n},\ldots,Q_{J_{n},n})$ is
\begin{displaymath}
\mathbb{P}[J_{n}=j_{n},\mathbf{Q}_{n}=\mathbf{q}_{n}]=\frac{1}{j_{n}!}{n\choose q_{1,n},\ldots,q_{j_{n},n}}\frac{\theta^{j_{n}}}{(\theta)_{(n)}}\prod_{i=1}^{j_{n}}(q_{i,n}-1)!,
\end{displaymath}
then the distribution of $\mathbf{X}_{m}$ from a Dirichlet process with atomic base measure $\theta\nu_{0}+\sum_{1\leq i\leq n}\delta_{Z^{\ast}_{i}}$ is
\begin{align}\label{dist_firstsample}
&\notag\mathbb{P}[\mathbf{N}_{m}=\mathbf{n}_{m},\mathbf{M}_{n}=\mathbf{m}_{m},K_{m}=k_{m},V_{m}=v_{m}]\\
&\notag\quad=\frac{\frac{\theta^{n+k_{m}}}{(\theta)_{(n+m)}}}{\frac{\theta^{n}}{(\theta)_{(n)}}\prod_{i=1}^{n}(1-1)!}\\
&\notag\quad\quad\times{m\choose v_{m}}{m-v_{m}\choose m_{1,m},\ldots,m_{n,m}}\prod_{i=1}^{n}(1+m_{i,m}-1)!\\
&\notag\quad\quad\quad\times\frac{1}{k_{m}!}{v_{m}\choose n_{1,m},\ldots,n_{k_{m},m}}\prod_{i=1}^{k_{m}}(n_{i,m}-1)!\\
&\quad=\frac{\theta^{k_{m}}}{(\theta+n)_{(m)}}{m\choose v_{m}}{m-v_{m}\choose m_{1,m},\ldots,m_{n,m}}\prod_{i=1}^{n}m_{i,m}!\\
&\quad\quad\times\frac{1}{k_{m}!}{v_{m}\choose n_{1,m},\ldots,n_{k_{m},m}}\prod_{i=1}^{k_{m}}(n_{i,m}-1)!.
\end{align}
In addition to the above preliminaries on the distribution of the random sample $\mathbf{X}_{m}$, for any $n\geq1$, $x\in\{1,\ldots,n\}$ and $1\leq\tau_{1}< \ldots <\tau_{x}\leq n$, let $\mathbf{M}_{(\tau_{1},\ldots,\tau_{x}),m}=(M_{\tau_{1},m},\ldots,M_{\tau_{x},m})$ be a collection of $x$ components of $\mathbf{M}_{m}$. We denote by $S(n,k)$ the Stirling number of the second kind, and by $\mathcal{C}_{n,x}$ the set of $x$-combinations without repetition of $\{1,\ldots,n\}$, i.e., $\mathcal{C}_{n,x}=\{(c_{1},\ldots,c_{x}):  c_{k}\in \{1,\ldots,n\}, c_{k}\neq c_{l},\text{ if } k\neq \l \}$ for any $x\geq1$, and $\mathcal{C}_{n,0}=\emptyset$. See Charalambides \cite{Cha(05)} for details.

\smallskip

\textsc{Proof of Theorem \ref{thm2}}. We start by determining the distribution of the random variable $R_{n,m}$, and then we show the distributional identity for $D_{m}(t)$. In order to compute the distribution of $R_{n,m}$, we start by  computing the corresponding $r$-th descending factorial moments. By the Vandermonde formula, we can write
\begin{align}\label{prior_freq_totali_merge}
\mathbb{E}[(R_{n,m})_{[r]}]=\sum_{s=0}^{r}{r\choose s}(-1)^{s}(n-s)_{[r-s]}\mathbb{E}[(R_{n,m}^{\ast})_{[s]}]
\end{align}
where $R^{\ast}_{n,m}=\sum_{1\leq i\leq n}\mathbbm{1}_{\{M_{i,m}=0\}}$. A repeated application of the Binomial theorem leads to write the $r$-th moment of the random variable $R^{\ast}_{n,m}$ as follows
\begin{align}\label{prior_freq_totali_1}
&\mathbb{E}[(R_{n,m}^{\ast})^{r}\,|\,V_{m}=v_{m},K_{m}=k_{m}]\\
&\notag\quad=\sum_{x=1}^{n}\sum_{i_{1}=1}^{r-1}\sum_{i_{2}=1}^{i_{1}-1}\cdots\sum_{i_{x-1}=1}^{i_{x-2}-1}{r\choose i_{1}}{i_{1}\choose i_{2}}\cdots{i_{x-2}\choose i_{x-1}}\\
&\notag\quad\quad\times\sum_{\mathbf{c}^{(x)}\in\mathcal{C}_{n,x}}\mathbb{E}\left[\prod_{t=1}^{x}(\mathbbm{1}_{\{M_{c_{t},m}=0\}})^{i_{x-t}-i_{x-t+1}}\,|\,V_{m}=v_{m},K_{m}=k_{m}\right]\\
&\notag\quad=\sum_{x=1}^{r}S(r,x)x!\sum_{\mathbf{c}^{(x)}\in\mathcal{C}_{n,x}}\mathbb{P}[\mathbf{M}_{\mathbf{c}_{x},m}=(\underbrace{0\,\ldots,0}_{x})\,|\,V_{m}=v_{m},K_{m}=k_{m}],
\end{align}
where
\begin{align}\label{prior_freq_totali_2}
&\mathbb{P}[\mathbf{M}_{\mathbf{c}_{x},m}=(\underbrace{0\,\ldots,0}_{x})\,|\,V_{m}=v_{m},K_{m}=k_{m}]=\frac{(n-x)_{(m-v_{m})}}{(n)_{(m-v_{m})}}, 
\end{align}
and
\begin{equation}\label{prior_freq_totali_3}
\mathbb{P}[V_{m}=v_{m}]=\frac{{m\choose v_{m}}}{(\theta+n)_{m}}(n)_{(m-v_{m})}(\theta)_{v_{m}}.
\end{equation}
Therefore, by combining Equation \eqref{prior_freq_totali_1} with Equation \eqref{prior_freq_totali_2} and Equation \eqref{prior_freq_totali_3} one has
\begin{displaymath}
\mathbb{E}[(R^{\ast}_{n,m})_{[r]}]=r!{n\choose r}\frac{(\theta+n-r)_{(m)}}{(\theta+n)_{(m)}},
\end{displaymath}
and from \eqref{prior_freq_totali_merge}
\begin{align}\label{prior_freq_totali_merge_distribut}
\mathbb{E}[(R_{n,m})_{[r]}]&=\sum_{s=0}^{r}{r\choose s}(-1)^{s}(n-s)_{[r-s]}s!{n\choose s}\frac{(\theta+n-s)_{(m)}}{(\theta+n)_{(m)}}\\
&\notag=\frac{r!}{(\theta+n)_{(m)}}\sum_{s=0}^{r}{n-s\choose r-s}(-1)^{s}{n\choose s}(\theta+n-s)_{(m)}.
\end{align}
The distribution of the random variable $R_{n,m}$ follows from the factorial moments in Equation \eqref{prior_freq_totali_merge_distribut}. In particular, for any $x=0,\ldots,\min(n,\lfloor m \rfloor)$, we can write the following
\begin{align*}
&\mathbb{P}[R_{n,m}=x]\\
&\notag\quad=\sum_{y\geq0}\frac{(-1)^{y}}{x!y!}\mathbb{E}[(R_{n,m})_{[x+y]}]\\
&\notag\quad=\frac{1}{(\theta+n)_{(m)}}\sum_{y\geq x}\frac{1}{x!}(-1)^{y-x}(y)_{[x]}\sum_{s=0}^{y}{n-s\choose y-s}(-1)^{s}{n\choose s}(\theta+n-s)_{(m)}\\
&\notag\quad=\frac{1}{(\theta+n)_{(m)}}\sum_{s=0}^{n}(-1)^{s}{n\choose s}(\theta+n-s)_{(m)}\sum_{y=s}^{n}(-1)^{y-x}{y\choose x}{n-s\choose y-s}\\
&\notag\quad=\frac{(-1)^{-x}}{(\theta+n)_{(m)}}\sum_{s=0}^{n}(-1)^{n-s}{n\choose s}{s\choose n-x}(\theta+n-s)_{(m)}\\
&\notag\quad=\frac{(-1)^{n}}{(\theta+n)_{(m)}}\sum_{s=n}^{n+x}(-1)^{s}{n\choose s-x}{s-x\choose n-x}(\theta+n-s+x)_{(m)}\\
&\notag\quad=\frac{{n\choose n-x}}{(\theta+n)_{(m)}}\sum_{s=0}^{x}(-1)^{s}{x\choose s}(\theta-s+x)_{(m)}\\
&\notag\quad=x!{n\choose x}{m\choose x}\frac{(\theta+x)_{(m-x)}}{(\theta+n)_{(m)}},
\end{align*}
where the last equality arises by an application of the Vandermonde identity. This proves \eqref{eq_prior_r}. As regards the distributional identity \eqref{eq:id_prior}, let us randomize the distribution of $R_{n,m}$ on $n$ with respect to the distribution \eqref{eq_lin_pop}. We can write
\begin{align}\label{solut_first}
&\sum_{n\geq x}d_{n}(t)x!{n\choose x}{m\choose x}\frac{(\theta+x)_{(m-x)}}{(\theta+n)_{(m)}}\\
&\notag\quad=\sum_{n\geq x}\sum_{i\geq n}\rho_{i}(t)(-1)^{-n}{i\choose n}(\theta+n)_{(i-1)}\frac{x!}{i!}{n\choose x}{m\choose x}\frac{(\theta+x)_{(m-x)}}{(\theta+n)_{(m)}}\\
&\notag\quad=\sum_{i\geq x}\sum_{n=x}^{i}\rho_{i}(t)(-1)^{-n}{i\choose n}(\theta+n)_{(i-1)}\frac{x!}{i!}{n\choose x}{m\choose x}\frac{(\theta+x)_{(m-x)}}{(\theta+n)_{(m)}}\\
&\notag\quad=\sum_{i\geq x}\rho_{i}(t)\frac{1}{i!}{m\choose x}x!(\theta+x)_{(m-x)}\sum_{n=x}^{i}{i\choose n}(-1)^{n}(\theta+n)_{(i-1)}\frac{{n\choose x}}{(\theta+n)_{(m)}}.
\end{align}
Let us focus on the second factor appearing in the last expression, namely the term $\sum_{n=x}^{i}{i\choose n}(-1)^{n}(\theta+n)_{(i-1)}{n\choose x}/(\theta+n)_{(m)}$. In particular we can rewrite it as
\begin{align}\label{solut_second}
&\frac{i!(-1)^{x}}{(i-x)!x!}\sum_{n=0}^{i-x}(-1)^{n}{i-x\choose n}\frac{(\theta+n+x)_{(i-1)}}{(\theta+n+x)_{(m)}}\\
&\notag\quad=\frac{i!(-1)^{x}}{(i-x)!x!(m-i)!}\sum_{n=0}^{i-x}(-1)^{n}{i-x\choose n}\int_{0}^{1}y^{\theta+n+x+i-1-1}(1-y)^{m-i+1-1}\ddr y\\
&\notag\quad=\frac{i!(-1)^{x}}{(i-x)!x!(m-i)!}\int_{0}^{1}y^{\theta+x+i-1-1}(1-y)^{m-x+1-1}\ddr y\\
&\notag\quad=\frac{i!(-1)^{x}}{(i-x)!x!(m-i)!}\frac{\Gamma(m-x+1)\Gamma(\theta+x+i-1)}{\Gamma(\theta+i+m)}.
\end{align}
Finally, by combining the expression \eqref{solut_first} with the expression \eqref{solut_second} one obtains what follows
\begin{align*}
&\sum_{n\geq x}d_{n}(t)x!{n\choose x}{m\choose x}\frac{(\theta+x)_{(m-x)}}{(\theta+n)_{(m)}}\\
&\quad=\sum_{i\geq x}\rho_{i}(t)\frac{1}{i!}{m\choose x}x!(\theta+x)_{(m-x)}\\
&\quad\quad\times\frac{i!(-1)^{x}}{(i-x)!x!(m-i)!}\frac{\Gamma(m-x+1)\Gamma(\theta+x+i-1)\Gamma(\theta+x)\Gamma(\theta+m)}{\Gamma(\theta+i+m)\Gamma(\theta+x)\Gamma(\theta+m)}\\
&\quad=\sum_{i\geq x}\rho_{i}(t)\frac{1}{i!}{m\choose x}x!\frac{i!(-1)^{x}}{(i-x)!x!(m-i)!}\frac{\Gamma(m-x+1)(\theta+x)_{(i-1)}}{(\theta+m)_{(i)}}\\
&\quad=\sum_{i=x}^{m}\rho_{i}(t){m\choose i}\frac{i!(-1)^{x}}{(i-x)!x!}\frac{(\theta+x)_{(i-1)}}{(\theta+m)_{(i)}},
\end{align*}
which, after some rearrangement of terms, coincides with the sampling ancestral distribution \eqref{eq_lin_sam}. This proves the distributional identity \eqref{eq:id_prior}, and the proof is completed.\hfill \qed

\smallskip

\textsc{Proof of Theorem \ref{thm1}}. The proof is along line similar to the first part of the proof of Theorem \ref{thm2}. We start by determining the $r$-th factorial moments of the random variable $R_{l,n,m}$. In particular, by a repeated application of the Binomial theorem
\begin{align}\label{prior_freq_1}
&\notag\mathbb{E}[(R_{l,n,m})^{r}\,|\,V_{m}=v_{m},K_{m}=k_{m}]\\
&\notag\quad=\sum_{x=1}^{n}\sum_{i_{1}=1}^{r-1}\sum_{i_{2}=1}^{i_{1}-1}\cdots\sum_{i_{x-1}=1}^{i_{x-2}-1}{r\choose i_{1}}{i_{1}\choose i_{2}}\cdots{i_{x-2}\choose i_{x-1}}\\
&\notag\quad\quad\times\sum_{\mathbf{c}^{(x)}\in\mathcal{C}_{n,x}}\mathbb{E}\left[\prod_{t=1}^{x}(\mathbbm{1}_{\{M_{c_{t},m}=l\}})^{i_{x-t}-i_{x-t+1}}\,|\,V_{m}=v_{m},K_{m}=k_{m}\right]\\
&\quad=\sum_{x=1}^{r}S(r,x)x!\sum_{\mathbf{c}^{(x)}\in\mathcal{C}_{n,x}}\mathbb{P}[\mathbf{M}_{\mathbf{c}_{x},m}=(\underbrace{l\,\ldots,l}_{x})\,|\,V_{m}=v_{m},K_{m}=k_{m}],
\end{align}
where 
\begin{align}\label{prior_freq_2}
&\mathbb{P}[\mathbf{M}_{\mathbf{c}_{x},m}=(\underbrace{l\,\ldots,l}_{x})\,|\,V_{m}=v_{m},K_{m}=k_{m}]=(xl)!\frac{{m-v_{m}\choose xl}(n-x)_{(m-v_{m}-xl)}}{(n)_{(m-v_{m})}}, 
\end{align}
and the distribution of the random variable $V_{m}$ is given in \eqref{prior_freq_totali_3}. Therefore, by a combination of Equation \eqref{prior_freq_1} with Equation \eqref{prior_freq_2} and Equation \eqref{prior_freq_totali_3} one obtains
\begin{align}\label{prior_freq_totali_3_distribut}
&\mathbb{E}[(R_{l,n,m})_{[r]}]=m!r!\frac{{n\choose r}}{(m-rl)!}\frac{(\theta+n-r)_{(m-rl)}}{(\theta+n)_{(m)}}.
\end{align}
Finally, the distribution of the random variable $R_{l,n,m}$ then follows from the factorial moments \eqref{prior_freq_totali_3_distribut}. In particular, for any $x=0,\ldots,\min(n,\lfloor m/l \rfloor)$, we can write
\begin{align}\label{prior_freq_4}
\mathbb{P}[R_{l,n,m}=x]&=\sum_{y\geq0}\frac{(-1)^{y}}{x!y!}\mathbb{E}[(R_{l,n,m})_{[x+y]}]\\
&\notag=\sum_{y\geq 0}(-1)^{y}\frac{1}{x!y!}(x+y)!{n\choose x+y}\\
&\notag\quad\times\frac{m!}{(m-(x+y)l)!}\frac{(\theta+n-x-y)_{(m-(x+y)l)}}{(\theta+n)_{(m)}}.
\end{align}
The expression \eqref{prior_freq_4} coincides, after some simplification and rearrangement of terms, to \eqref{descent_freq}. In particular, the sum over the index $y$ ranges between $x$ and $\min(n,\lfloor m/l \rfloor)$, where $\lfloor m/l \rfloor$ denotes the integer part of $m/l$. The proof is completed.\hfill \qed

\subsection{Proofs of Section \ref{subsec2}}\label{appendix2}

Let $\mathbf{X}_{m}$ be a random sample from a Dirichlet process with base measure $\theta\nu_{0}+\sum_{1\leq i\leq n}\delta_{Z^{\ast}_{i}}$. Recall that $\mathbf{X}_{n}$ may be viewed as a a random sample from a posterior Dirichlet process given the sample $(Z_{1}^{\ast},\ldots,Z^{\ast})$ featuring $n$ distinct types, and that the distribution of $\mathbf{X}_{n}$ is given in Equation \eqref{dist_firstsample}. We start by describing the composition of an additional sample $\mathbf{X}_{m^{\prime}}$, for $m^{\prime}\geq0$. Let $\{X_{K_{m}+1}^{\ast},\ldots,X_{K_{m}+K_{m^{\prime}}}^{\ast}\}$ be the labels identifying the $K_{m^{\prime}}$ distinct types in the sample $\mathbf{X}_{m^{\prime}}$ which do not coincide with any of $\{Z_{1}^{\ast},\ldots,Z_{n}^{\ast},X_{1}^{\ast},\ldots,X_{K_{m}}^{\ast}\}$. Moreover, let 
\begin{itemize}
\item[i)] $M_{j,m^{\prime}}=\sum_{1\leq i\leq m^{\prime}}\mathbbm{1}_{\{Z^{\ast}_{j}\}}(X_{m+i})$ be the number of $X_{m+i}$'s that coincide with $Z^{\ast}_{j}$, for any $j=1,\ldots,n$,
\item[ii)] $N_{j,m^{\prime}}=\sum_{1\leq i\leq m^{\prime}}\mathbbm{1}_{\{X^{\ast}_{j}\}}(X_{m+i})$ be the number of $X_{m+i}$'s that coincide with $X^{\ast}_{j}$, for any $j=1,\ldots,K_{m}+K_{m^{\prime}}$.
\end{itemize}
Additionally, let $\mathbf{N}_{m^{\prime}}=(N_{1,m^{\prime}},\ldots,N_{K_{m},m^{\prime}},N_{K_{m}+1,m^{\prime}},\ldots,N_{K_{m}+K_{m^{\prime}},m^{\prime}})$ and $\mathbf{M}_{m^{\prime}}=(M_{1,m^{\prime}},\ldots,M_{n,m^{\prime}})$. Also, we denote by $V_{m^{\prime}}$ the number of $X_{m+i}$'s which do not coincide with any of the labels $\{Z_{1}^{\ast},\ldots,Z_{n}^{\ast},X_{1}^{\ast},\ldots,X_{K_{m}}^{\ast}\}$, i.e., we can write
\begin{displaymath}
V_{m^{\prime}}=\sum_{i=1}^{K_{m^{\prime}}}N_{K_{m}+i,m^{\prime}}.
\end{displaymath}
In a similar way, we denote by $W_{m^{\prime}}$ the number of $X_{m+i}$'s which do not coincide with any of the labels $\{Z_{1}^{\ast},\ldots,Z_{n}^{\ast},X_{K_{m}+1}^{\ast},\ldots,X_{K_{m}+K_{m^{\prime}}}^{\ast}\}$, i.e., we can write $W_{m^{\prime}}$ as
\begin{displaymath}
W_{m^{\prime}}=\sum_{i=1}^{K_{m}}N_{i,m^{\prime}}.
\end{displaymath}
We can write the conditional probability of $(\mathbf{N}_{m^{\prime}},\mathbf{M}_{m^{\prime}},V_{m^{\prime}}, W_{m^{\prime}}, K_{m^{\prime}})$ given $(\mathbf{N}_{m},\mathbf{M}_{m},K_{m})$, where $\mathbf{N}_{m}$, $\mathbf{M}_{m}$, and $K_{m}$ have been defined in Section \ref{sec2}. This may be viewed as the natural conditional (posterior) counterpart of \eqref{dist_firstsample}. In particular, by a direct application of Equation \eqref{dist_firstsample}, we can write the following probability
\begin{align}\label{global_post}
&\mathbb{P}[\mathbf{N}_{m^{\prime}}=\mathbf{n}_{m^{\prime}},\mathbf{M}_{m^{\prime}}=\mathbf{m}_{m^{\prime}},V_{m^{\prime}}=v_{m^{\prime}},W_{m^{\prime}}=w_{m^{\prime}},K_{m^{\prime}}=k_{m^{\prime}}\\
&\notag\quad\quad\quad\quad\quad\quad\quad\quad\quad\quad\quad\quad\quad\quad\quad|\,\mathbf{N}_{m}=\mathbf{n}_{m},\mathbf{M}_{m}=\mathbf{m}_{m},K_{m}=k_{m}]\\
&\notag\quad=\frac{\frac{\theta^{n+k_{m}+k_{m^{\prime}}}}{(\theta)_{(n+m+m^{\prime})}}}{\frac{\theta^{n+k_{m}}}{(\theta)_{(n+m)}}\prod_{i=1}^{n}(1+m_{i,m}-1)!\prod_{i=1}^{k_{m}}(n_{i,m}-1)!}\\
&\notag\quad\quad\times{m^{\prime}\choose v_{m^{\prime}},w_{m^{\prime}},m^{\prime}-v_{m^{\prime}}-w_{m^{\prime}}}\\
&\notag\quad\quad\quad\times{m^{\prime}-v_{m^{\prime}}-w_{m^{\prime}}\choose m_{1,m^{\prime}},\ldots,m_{n,m^{\prime}}}\prod_{i=1}^{n}(1+m_{i,m}+m_{i,m^{\prime}}-1)!\\
&\notag\quad\quad\quad\quad\times{w_{m^{\prime}}\choose n_{1,m^{\prime}},\ldots,n_{k_{n},m^{\prime}}}\prod_{i=1}^{k_{m}}(n_{i,m}+n_{i,m^{\prime}}-1)!\\
&\notag\quad\quad\quad\quad\quad\times\frac{1}{k_{m^{\prime}}!}{v_{m^{\prime}}\choose n_{k_{n}+1,m^{\prime}},\ldots,n_{k_{n}+k_{m^{\prime}},m^{\prime}}}\prod_{i=1}^{k_{m^{\prime}}}(n_{k_{m}+i,m^{\prime}}-1)!\\
&\notag\quad=\frac{\frac{\theta^{k_{m^{\prime}}}}{(\theta+n+m)_{(m^{\prime})}}}{\prod_{i=1}^{n}m_{i,m}!\prod_{i=1}^{k_{m}}(n_{i,m}-1)!}\\
&\notag\quad\quad\times{m^{\prime}\choose v_{m^{\prime}},w_{m^{\prime}},m^{\prime}-v_{m^{\prime}}-w_{m^{\prime}}}\\
&\notag\quad\quad\quad\times{m^{\prime}-v_{m^{\prime}}-w_{m^{\prime}}\choose m_{1,m^{\prime}},\ldots,m_{n,m^{\prime}}}\prod_{i=1}^{n}(m_{i,m}+m_{i,m^{\prime}})!\\
&\notag\quad\quad\quad\quad\times{w_{m^{\prime}}\choose n_{1,m^{\prime}},\ldots,n_{k_{n},m^{\prime}}}\prod_{i=1}^{k_{m}}(n_{i,m}+n_{i,m^{\prime}}-1)!\\
&\notag\quad\quad\quad\quad\quad\times\frac{1}{k_{m^{\prime}}!}{v_{m^{\prime}}\choose n_{k_{n}+1,m^{\prime}},\ldots,n_{k_{n}+k_{m^{\prime}},m^{\prime}}}\prod_{i=1}^{k_{m^{\prime}}}(n_{k_{m}+i,m^{\prime}}-1)!
\end{align}
To simplify the notation we define $A_{m}(\mathbf{n}_{m},\mathbf{m}_{m},k_{m})=\{\mathbf{N}_{m}=\mathbf{n}_{m},\mathbf{M}_{m}=\mathbf{m}_{m},K_{m}=k_{m}\}$ and $B_{m^{\prime}}(v_{m^{\prime}},w_{m^{\prime}},k_{m^{\prime}})=\{V_{m^{\prime}}=v_{m^{\prime}},W_{m^{\prime}}=w_{m^{\prime}},K_{m^{\prime}}=k_{m^{\prime}}\}$. Furthermore, with a slight abuse of notation, we denote by $X\,|\,Y$ a random variable whose distribution coincides with the conditional distribution of $X$ given $Y$. 

\begin{lem}
For any $n\geq1$, $x\in\{1,\ldots,n\}$ and $1\leq\tau_{1}< \ldots,<\tau_{x}\leq n$, let $\mathbf{M}_{(\tau_{1},\ldots,\tau_{x}),m^{\prime}}=(M_{\tau_{1},m^{\prime}},\ldots,M_{\tau_{x},m^{\prime}})$ be a collection of $x$ components of $\mathbf{M}_{m^{\prime}}$. Then
\begin{align}\label{eq:lem1}
&\mathbb{P}[\mathbf{M}_{(\tau_{1},\ldots,\tau_{x}),m^{\prime}}=\mathbf{m}_{(\tau_{1},\ldots,\tau_{x}),m^{\prime}}\,|\,A_{m}(\mathbf{n}_{m},\mathbf{m}_{m},k_{m}),B_{m^{\prime}}(v_{m^{\prime}},w_{m^{\prime}},k_{m^{\prime}})]\\
&\notag\quad= {m^{\prime}-v_{m^{\prime}}-w_{m^{\prime}}\choose m_{\tau_{1},m^{\prime}},\ldots,m_{\tau_{x},m^{\prime}},m^{\prime}-v_{m^{\prime}}-w_{m^{\prime}}-\sum_{i=1}^{x}m_{\tau_{i},m^{\prime}}}\\
&\notag\quad\quad\times\frac{(n+m-\sum_{i=1}^{k_{m}}n_{i,m}-\sum_{i=1}^{x}(1+m_{\tau_{i},m}))_{(m^{\prime}-v_{m^{\prime}}-w_{m^{\prime}}-\sum_{i=1}^{x}m_{\tau_{i},m^{\prime}})}}{(n+m-\sum_{i=1}^{k_{m}}n_{i,m})_{(m^{\prime}-v_{m^{\prime}}-w_{m^{\prime}})}}\\
&\notag\quad\quad\quad\times\prod_{i=1}^{x}(1+m_{\tau_{i},m})_{(m_{\tau_{i},m^{\prime}})}.
\end{align}
\end{lem}
\begin{proof}
We start by determining the conditional distribution of the random variable $(V_{m^{\prime}}, W_{m^{\prime}}, K_{m^{\prime}})$ given the sample $\mathbf{X}_{m}$. This is obtained by suitably marginalizing the distribution \eqref{global_post} over $(\mathbf{N}_{m^{\prime}},\mathbf{M}_{m^{\prime}})$. In particular, with this regards, if
\begin{align*}
&\mathcal{S}^{(0)}_{m^{\prime}-v_{m^{\prime}}-w_{m^{\prime}},n}=\left\{(m_{i,m^{\prime}})_{1\leq i\leq n}\text{ : }m_{i,m^{\prime}}\geq0\wedge\sum_{i=1}^{n}m_{i,m^{\prime}}=m^{\prime}-v_{m^{\prime}}-w_{m^{\prime}}\right\},
\end{align*}
\begin{displaymath}
\mathcal{S}^{(0)}_{w_{m^{\prime}},k_{m}}=\left\{(n_{i,m^{\prime}})_{1\leq i\leq k_{m}}\text{ : }n_{i,m^{\prime}}\geq0\wedge\sum_{i=1}^{k_{m}}n_{i,m^{\prime}}=w_{m^{\prime}}\right\},
\end{displaymath}
and
\begin{align*}
&\mathcal{S}_{v_{m^{\prime}},k_{m^{\prime}}}=\left\{(n_{k_{m}+i,m^{\prime}})_{1\leq i\leq k_{m^{\prime}}}\text{ : }n_{k_{m}+i,m^{\prime}}\geq1\wedge\sum_{i=1}^{k_{m^{\prime}}}n_{k_{n}+i,m^{\prime}}=v_{m^{\prime}}\right\},
\end{align*}
then
\begin{align}\label{marginal_vwkzp}
&\mathbb{P}[V_{m^{\prime}}=v_{m^{\prime}},W_{m^{\prime}}=w_{m^{\prime}},K_{m^{\prime}}=k_{m^{\prime}}\,|\,A_{m}(\mathbf{n}_{m},\mathbf{m}_{m},k_{m})]\\
&\notag\quad=\frac{\frac{\theta^{n+k_{m}+k_{m^{\prime}}}}{(\theta)_{(n+m+m^{\prime})}}}{\frac{\theta^{n+k_{m}}}{(\theta)_{(n+m)}}\prod_{i=1}^{n}m_{i,m}!\prod_{i=1}^{k_{m}}(n_{i,m}-1)!}\\
&\notag\quad\quad\times{m^{\prime}\choose v_{m^{\prime}},w_{m^{\prime}},m^{\prime}-v_{m^{\prime}}-w_{m^{\prime}}}\\
&\notag\quad\quad\quad\times\sum_{\mathcal{S}^{(0)}_{m^{\prime}-v_{m^{\prime}}-w_{m^{\prime}},n}}{m^{\prime}-v_{m^{\prime}}-w_{m^{\prime}}\choose m_{1,m^{\prime}},\ldots,m_{n,m^{\prime}}}\prod_{i=1}^{n}(m_{i,m}+m_{i,m^{\prime}})!\\
&\notag\quad\quad\quad\quad\times\sum_{\mathcal{S}^{(0)}_{w_{m^{\prime}},k_{m}}}{w_{m^{\prime}}\choose n_{1,m^{\prime}},\ldots,n_{k_{m},m^{\prime}}}\prod_{i=1}^{k_{m}}(n_{i,m}+n_{i,m^{\prime}}-1)!\\
&\notag\quad\quad\quad\quad\quad\times\frac{1}{k_{m^{\prime}}!}\sum_{\mathcal{S}_{v_{m^{\prime}},k_{m^{\prime}}}}{v_{m^{\prime}}\choose n_{k_{m}+1,m^{\prime}},\ldots,n_{k_{m}+k_{m^{\prime}},m^{\prime}}}\prod_{i=1}^{k_{m^{\prime}}}(n_{k_{m}+i,m^{\prime}}-1)!
\end{align}
Now, we apply Vandermonde formula and Theorem 2.5 in Charalambides \cite{Cha(05)} in order to solve the above summations. In particular, we have the following identities
\begin{itemize}
\item[i)]
\begin{align}\label{marginal_vwk_1}
&\sum_{\mathcal{S}^{(0)}_{m^{\prime}-v_{m^{\prime}}-w_{m^{\prime}},n}}{m^{\prime}-v_{m^{\prime}}-w_{m^{\prime}}\choose m_{1,m^{\prime}},\ldots,m_{n,m^{\prime}}}\prod_{i=1}^{n}(m_{i,m}+m_{i,m^{\prime}})!\\
&\notag\quad=(n+m-\sum_{i=1}^{k_{m}}n_{i,m})_{(m^{\prime}-v_{m^{\prime}}-w_{m^{\prime}})}\prod_{i=1}^{n}m_{i,m}!
\end{align}
\item[ii)]
\begin{align}\label{marginal_vwk_2}
&\sum_{\mathcal{S}^{(0)}_{w_{m^{\prime}},k_{m}}}{w_{m^{\prime}}\choose n_{1,m^{\prime}},\ldots,n_{k_{m},m^{\prime}}}\prod_{i=1}^{k_{m}}(n_{i,m}+n_{i,m^{\prime}}-1)!\\
&\notag\quad=\left(\sum_{i=1}^{k_{m}}n_{i,m}\right)_{(w_{m^{\prime}})}\prod_{i=1}^{k_{m}}(n_{i,m}-1)!
\end{align}
\item[iii)]
\begin{align}\label{marginal_vwk_3}
&\frac{1}{k_{m^{\prime}}!}\sum_{\mathcal{S}_{v_{m^{\prime}},k_{m^{\prime}}}}{v_{m^{\prime}}\choose n_{k_{m}+1,m^{\prime}},\ldots,n_{k_{m}+k_{m^{\prime}},m^{\prime}}}\prod_{i=1}^{k_{m^{\prime}}}(n_{k_{m}+i,m^{\prime}}-1)!\\
&\notag\quad=|s(v_{m^{\prime}},k_{m^{\prime}})|
\end{align}
\end{itemize} where $|s(n,k)|$ denotes the signless Stirling number of the first type (see Charalambides \cite{Cha(05)}). By combining \eqref{marginal_vwkzp} with identities \eqref{marginal_vwk_1}, \eqref{marginal_vwk_2} and \eqref{marginal_vwk_3} we obtain 
\begin{align}\label{marginal_vwk}
&\notag\mathbb{P}[V_{m^{\prime}}=v_{m^{\prime}},W_{m^{\prime}}=w_{m^{\prime}},K_{m^{\prime}}=k_{m^{\prime}}\,|\,A_{m}(\mathbf{n}_{m},\mathbf{m}_{m},k_{m})]\\
&\notag\quad=\frac{\frac{\theta^{n+k_{m}+k_{m^{\prime}}}}{(\theta)_{(n+m+m^{\prime})}}}{\frac{\theta^{n+k_{m}}}{(\theta)_{(n+m)}}\prod_{i=1}^{n}m_{i,m}!\prod_{i=1}^{k_{m}}(n_{i,m}-1)!}\\
&\notag\quad\quad\times{m^{\prime}\choose v_{m^{\prime}},w_{m^{\prime}},m^{\prime}-v_{m^{\prime}}-w_{m^{\prime}}}\\
&\notag\quad\quad\quad\times(n+m-\sum_{i=1}^{k_{m}}n_{i,m})_{(m^{\prime}-v_{m^{\prime}}-w_{m^{\prime}})}\prod_{i=1}^{n}m_{i,m}!\\
&\notag\quad\quad\quad\quad\times\left(\sum_{i=1}^{k_{m}}n_{i,m}\right)_{(w_{m^{\prime}})}\prod_{i=1}^{k_{m}}(n_{i,m}-1)!|s(v_{m^{\prime}},k_{m^{\prime}})|\\
&\quad=\frac{\theta^{k_{m^{\prime}}}}{(\theta+n+m)_{(m^{\prime})}}{m^{\prime}\choose v_{m^{\prime}},w_{m^{\prime}},m^{\prime}-v_{m^{\prime}}-w_{m^{\prime}}}\\
&\notag\quad\quad\times\left(n+m-\sum_{i=1}^{k_{m}}n_{i,m}\right)_{(m^{\prime}-v_{m^{\prime}}-w_{m^{\prime}})}\left(\sum_{i=1}^{k_{m}}n_{i,m}\right)_{(w_{m^{\prime}})}|s(v_{m^{\prime}},k_{m^{\prime}})|,
\end{align}
Accordingly, from the probability \eqref{marginal_vwk} we can write the following marginal probability
\begin{align}\label{marginal_vw}
&\notag\mathbb{P}[V_{m^{\prime}}=v_{m^{\prime}},W_{m^{\prime}}=w_{m^{\prime}}\,|\,A_{m}(\mathbf{n}_{m},\mathbf{m}_{m},k_{m})]\\
&\notag\quad=\sum_{k_{m^{\prime}}=0}^{v_{m^{\prime}}}\frac{\theta^{k_{m^{\prime}}}}{(\theta+n+m)_{(m^{\prime})}}{m^{\prime}\choose v_{m^{\prime}},w_{m^{\prime}},m^{\prime}-v_{m^{\prime}}-w_{m^{\prime}}}\\
&\notag\quad\quad\times(n+m-\sum_{i=1}^{k_{m}}n_{i,m})_{(m^{\prime}-v_{m^{\prime}}-w_{m^{\prime}})}\left(\sum_{i=1}^{k_{m}}n_{i,m}\right)_{(w_{m^{\prime}})}|s(v_{m^{\prime}},k_{m^{\prime}})|\\
&\quad=\frac{1}{(\theta+n+m)_{(m^{\prime})}}{m^{\prime}\choose v_{m^{\prime}},w_{m^{\prime}},m^{\prime}-v_{m^{\prime}}-w_{m^{\prime}}}\\
&\notag\quad\quad\times\left(n+m-\sum_{i=1}^{k_{m}}n_{i,m}\right)_{(m^{\prime}-v_{m^{\prime}}-w_{m^{\prime}})}\left(\sum_{i=1}^{k_{m}}n_{i,m}\right)_{(w_{m^{\prime}})}(\theta)_{(v_{m^{\prime}})}.
\end{align}
By combining \eqref{global_post} with \eqref{marginal_vwk} one obtains the conditional distribution of the random variable $(\mathbf{N}_{m^{\prime}},\mathbf{M}_{m^{\prime}})$ given $(\mathbf{N}_{m},\mathbf{M}_{m},K_{m},V_{m^{\prime}},W_{m^{\prime}},K_{m^{\prime}})$. In particular,
\begin{align}\label{eq:modified like}
&\notag\mathbb{P}[\mathbf{N}_{m^{\prime}}=\mathbf{n}_{m^{\prime}},\mathbf{M}_{m^{\prime}}=\mathbf{m}_{m^{\prime}}\,|\,A_{m}(\mathbf{n}_{m},\mathbf{m}_{m},k_{m}),B_{m^{\prime}}(v_{m^{\prime}},w_{m^{\prime}},k_{m^{\prime}})]\\
&\notag\quad=\Bigg[\frac{\theta^{k_{m^{\prime}}}}{(\theta+n+m)_{(m^{\prime})}}{m^{\prime}\choose v_{m^{\prime}},w_{m^{\prime}},m^{\prime}-v_{m^{\prime}}-w_{m^{\prime}}}\\
&\notag\quad\quad\times\left(n+m-\sum_{i=1}^{k_{m}}n_{i,m}\right)_{(m^{\prime}-v_{m^{\prime}}-w_{m^{\prime}})}\left(\sum_{i=1}^{k_{m}}n_{i,m}\right)_{(w_{m^{\prime}})}|s(v_{m^{\prime}},k_{m^{\prime}})|\Bigg]^{-1}\\
&\notag\quad\quad\quad\times\frac{\frac{\theta^{n+k_{m}+k_{m^{\prime}}}}{(\theta)_{(n+m+m^{\prime})}}}{\frac{\theta^{n+k_{m}}}{(\theta)_{(n+m)}}\prod_{i=1}^{n}m_{i,m}!\prod_{i=1}^{k_{m}}(n_{i,m}-1)!}\\
&\notag\quad\quad\quad\quad\times{m^{\prime}\choose v_{m^{\prime}},w_{m^{\prime}},m^{\prime}-v_{m^{\prime}}-w_{m^{\prime}}}\\
&\notag\quad\quad\quad\quad\quad\times{m^{\prime}-v_{m^{\prime}}-w_{m^{\prime}}\choose m_{1,m^{\prime}},\ldots,m_{n,m^{\prime}}}\prod_{i=1}^{n}(m_{i,m}+m_{i,m^{\prime}})!\\
&\notag\quad\quad\quad\quad\quad\quad\times{w_{m^{\prime}}\choose n_{1,m^{\prime}},\ldots,n_{k_{n},m^{\prime}}}\prod_{i=1}^{k_{m}}(n_{i,m}+n_{i,m^{\prime}}-1)!\\
&\notag\quad\quad\quad\quad\quad\quad\quad\times\frac{1}{k_{m^{\prime}}!}{v_{m^{\prime}}\choose n_{k_{n}+1,m^{\prime}},\ldots,n_{k_{n}+k_{m^{\prime}},m^{\prime}}}\prod_{i=1}^{k_{m^{\prime}}}(n_{k_{m}+i}-1)!\\
&\quad=\frac{\left[(n+m-\sum_{i=1}^{k_{m}}n_{i,m})_{(m^{\prime}-v_{m^{\prime}}-w_{m^{\prime}})}(\sum_{i=1}^{k_{m}}n_{i,m})_{w_{m^{\prime}}}|s(v_{m^{\prime}},k_{m^{\prime}})|\right]^{-1}}{\prod_{i=1}^{n}m_{i,m}!\prod_{i=1}^{k_{m}}(n_{i,m}-1)!}\\
&\notag\quad\quad\times{m^{\prime}-v_{m^{\prime}}-w_{m^{\prime}}\choose m_{1,m^{\prime}},\ldots,m_{n,m^{\prime}}}\prod_{i=1}^{n}(m_{i,m}+m_{i,m^{\prime}})!\\
&\notag\quad\quad\quad\times{w_{m^{\prime}}\choose n_{1,m^{\prime}},\ldots,n_{k_{n},m^{\prime}}}\prod_{i=1}^{k_{m}}(n_{i,m}+n_{i,m^{\prime}}-1)!\\
&\notag\quad\quad\quad\quad\times\frac{1}{k_{m^{\prime}}!}{v_{m^{\prime}}\choose n_{k_{n}+1,m^{\prime}},\ldots,n_{k_{n}+k_{m^{\prime}},m^{\prime}}}\prod_{i=1}^{k_{m^{\prime}}}(n_{k_{m}+i}-1)!.
\end{align}
The distribution \eqref{eq:modified like} leads to the conditional distribution \eqref{eq:lem1}. For any $x\in\{1,\ldots,n\}$, let $1\leq\tau_{1}< \ldots <\tau_{x}\leq n$ let $\mathcal{J}_{n,x}=\{0,\ldots,n\}/\{\tau_{1},\ldots,\tau_{x}\}$. Also, let
\begin{align*}
&\mathcal{S}^{(0)}_{m^{\prime}-v_{m^{\prime}}-w_{m^{\prime}}-\sum_{i=1}^{x}m_{\tau_{i},m^{\prime}},n-x}\\
&=\left\{(m_{i,m^{\prime}})_{i\in\mathcal{J}_{n,x}} :m_{i,m^{\prime}}\geq0\wedge\sum_{i\in\mathcal{J}_{n,x}}m_{i,m^{\prime}}=m^{\prime}-v_{m^{\prime}}-w_{m^{\prime}}-\sum_{i=1}^{x}m_{\tau_{i},m^{\prime}}\right\}.
\end{align*}
Marginalizing \eqref{eq:modified like} over $\mathcal{S}^{(0)}_{m^{\prime}-v_{m^{\prime}}-w_{m^{\prime}}-\sum_{i=1}^{x}m_{\tau_{i},m^{\prime}},n-x}$, $\mathcal{S}^{(0)}_{w_{m^{\prime}},k_{m}}$ and $\mathcal{S}_{v_{m^{\prime}},k_{m^{\prime}}}$ one obtains, by simple algebraic manipulations, the marginal conditional distribution
\begin{align}\label{eq:final}
&\notag\mathbb{P}[\mathbf{M}_{(\tau_{1},\ldots,\tau_{x}),m^{\prime}}=\mathbf{m}_{(\tau_{1},\ldots,\tau_{x}),m^{\prime}}\,|\,A_{m}(\mathbf{n}_{m},\mathbf{m}_{m},k_{m}),B_{m^{\prime}}(v_{m^{\prime}},w_{m^{\prime}},k_{m^{\prime}})]\\
&\notag\quad=\frac{\left[(n+m-\sum_{i=1}^{k_{m}}n_{i,m})_{(m^{\prime}-v_{m^{\prime}}-w_{m^{\prime}})}(\sum_{i=1}^{k_{m}}n_{i,m})_{(w_{m^{\prime}})}|s(v_{m^{\prime}},k_{m^{\prime}})|\right]^{-1}}{\prod_{i=1}^{n}m_{i,m}! \prod_{i=1}^{k_{m}}(n_{i,m}-1)!}\\
&\notag\quad\quad\times\sum_{\mathcal{S}^{(0)}_{m^{\prime}-v_{m^{\prime}}-w_{m^{\prime}}-\sum_{i=1}^{x}m_{\tau_{i},m^{\prime}},n-x}}{m^{\prime}-v_{m^{\prime}}-w_{m^{\prime}}\choose m_{1,m^{\prime}},\ldots,m_{n,m^{\prime}}}\prod_{i=1}^{n}(m_{i,m}+m_{i,m^{\prime}})!\\
&\notag\quad\quad\quad\times\sum_{\mathcal{S}^{(0)}_{w_{m^{\prime}},k_{m}}}{w_{m^{\prime}}\choose n_{1,m^{\prime}},\ldots,n_{k_{n},m^{\prime}}}\prod_{i=1}^{k_{m}}(n_{i,m}+n_{i,m^{\prime}}-1)!  \\
&\notag\quad\quad\quad\quad\times\frac{1}{k_{m^{\prime}}!}\sum_{\mathcal{S}_{v_{m^{\prime}},k_{m^{\prime}}}}{v_{m^{\prime}}\choose n_{k_{n}+1,m^{\prime}},\ldots,n_{k_{n}+k_{m^{\prime}},m^{\prime}}}\prod_{i=1}^{k_{m^{\prime}}}(n_{k_{n}+i,m^{\prime}}-1)!\\
&\notag\quad=\frac{\left[(n+m-\sum_{i=1}^{k_{m}}n_{i,m})_{(m^{\prime}-v_{m^{\prime}}-w_{m^{\prime}})}\right]^{-1}}{\prod_{i=1}^{n}m_{i,m}!}\\
&\notag\quad\quad\times\sum_{\mathcal{S}^{(0)}_{m^{\prime}-v_{m^{\prime}}-w_{m^{\prime}}-\sum_{i=1}^{x}m_{\tau_{i},m^{\prime}},n-x}}{m^{\prime}-v_{m^{\prime}}-w_{m^{\prime}}\choose m_{1,m^{\prime}},\ldots,m_{n,m^{\prime}}}\prod_{i=1}^{n}(m_{i,m}+m_{i,m^{\prime}})!\\
&\notag\quad=\frac{\left[(n+m-\sum_{i=1}^{k_{m}}n_{i,m})_{(m^{\prime}-v_{m^{\prime}}-w_{m^{\prime}})}\right]^{-1}}{\prod_{i=1}^{n}m_{i,m}!}\\
&\notag\quad\quad\times\frac{(m^{\prime}-v_{m^{\prime}}-w_{m^{\prime}})!\prod_{i=1}^{x}(m_{\tau_{i},m}+m_{\tau_{i},m^{\prime}})!}{(m^{\prime}-v_{m^{\prime}}-w_{m^{\prime}}-\sum_{i=1}^{x}m_{\tau_{i},m^{\prime}})!\prod_{i=1}^{x}m_{\tau_{i},m^{\prime}}!}\\
&\notag\quad\quad\quad\times\sum_{\mathcal{S}^{(0)}_{m^{\prime}-v_{m^{\prime}}-w_{m^{\prime}}-\sum_{i=1}^{x}m_{\tau_{i},m^{\prime}},n-x}}\frac{(m^{\prime}-v_{m^{\prime}}-w_{m^{\prime}}-\sum_{i=1}^{x}m_{\tau_{i},m^{\prime}})!}{\prod_{i\in\mathcal{J}_{n,x}}m_{i,m^{\prime}}!}\\
&\notag\quad\quad\quad\quad\times\prod_{i\in\mathcal{J}_{n,x}}(1+m_{i,m}+m_{i,m^{\prime}}-1)!\\
&\quad=\frac{\left[(n+m-\sum_{i=1}^{k_{m}}n_{i,m})_{(m^{\prime}-v_{m^{\prime}}-w_{m^{\prime}})}\right]^{-1}}{\prod_{i=1}^{n}m_{i,m}!}\\
&\notag\quad\quad\times\frac{(m^{\prime}-v_{m^{\prime}}-w_{m^{\prime}})!\prod_{i=1}^{x}(m_{\tau_{i},m}+m_{\tau_{i},m^{\prime}})!}{(m^{\prime}-v_{m^{\prime}}-w_{m^{\prime}}-\sum_{i=1}^{x}m_{\tau_{i},m^{\prime}})!\prod_{i=1}^{x}m_{\tau_{i},m^{\prime}}!}\\
&\notag\quad\quad\quad\times\left(n+m-\sum_{i=1}^{k_{m}}n_{i,m}-\sum_{i=1}^{x}(1+m_{\tau_{i},m})\right)_{(m^{\prime}-v_{m^{\prime}}-w_{m^{\prime}}-\sum_{i=1}^{x}m_{\tau_{i},m^{\prime}})}\\
&\notag\quad\quad\quad\quad\times\prod_{i\in\mathcal{J}_{n,x}}m_{i,m}!,
\end{align}
follows by a direct application of Theorem 2.5 in Charalambides \cite{Cha(05)}. The expression \eqref{eq:final} coincides with the conditional distribution of $\mathbf{M}_{m^{(\tau_{1},\ldots,\tau_{x}),m^{\prime}}}$ given $(\mathbf{N}_{m},\mathbf{M}_{m},K_{m},V_{m^{\prime}},W_{m^{\prime}},K_{m^{\prime}})$ displayed in \eqref{eq:lem1}, and the proof is completed.
\end{proof}

\smallskip

\textsc{Proof of i) Theorem \ref{prp1}}. We compute the $r$-th descending factorial moment of the random variable $R_{n,m+m^{\prime}}$ given $(\mathbf{N}_{m},\mathbf{M}_{m},K_{m})$. In particular, by a direct application of the Vandermonde formula, we can write the following identity 
\begin{align*}
&\mathbb{E}[(R_{n,m+m^{\prime}})_{[r]}\,|\,A_{m}(\mathbf{n}_{m},\mathbf{m}_{m},k_{m}),B_{m^{\prime}}(v_{m^{\prime}},w_{m^{\prime}},k_{m^{\prime}})]\\
&\quad=\sum_{s=0}^{r}{r\choose s}(-1)^{s}(n-s)_{[r-s]}\mathbb{E}[(R_{n,m+m^{\prime}}^{\ast})_{[s]}\,|\,A_{m}(\mathbf{n}_{m},\mathbf{m}_{m},k_{m}),B_{m^{\prime}}(v_{m^{\prime}},w_{m^{\prime}},k_{m^{\prime}})],
\end{align*}
where $R^{\ast}_{n,m+m^{\prime}}=\sum_{1\leq i\leq n}\mathbbm{1}_{\{M_{i,m}+M_{i,m^{\prime}}=0\}}$. A repeated application of the Binomial theorem leads to write the $r$-th moment of the random variable $R^{\ast}_{n,m+m^{\prime}}$ as follows
\begin{align*}
&\notag\mathbb{E}[(R^{\ast}_{n,m+m^{\prime}})^{r}\,|\,A_{m}(\mathbf{n}_{m},\mathbf{m}_{m},k_{m}),B_{m^{\prime}}(v_{m^{\prime}},w_{m^{\prime}},k_{m^{\prime}})]\\
&\notag\quad=\sum_{x=1}^{n}\sum_{i_{1}=1}^{r-1}\sum_{i_{2}=1}^{i_{1}-1}\cdots\sum_{i_{x-1}=1}^{i_{x-2}-1}{r\choose i_{1}}{i_{1}\choose i_{2}}\cdots{i_{x-2}\choose i_{x-1}}\\
&\notag\quad\quad\times\sum_{\mathbf{c}^{(x)}\in\mathcal{C}_{n,x}}\mathbb{E}\left[\prod_{t=1}^{x}(\mathbbm{1}_{\{M_{c_{t},m}+M_{c_{t},m^{\prime}}=0\}})^{i_{x-t}-i_{x-t+1}}\,|\,A_{m}(\mathbf{n}_{m},\mathbf{m}_{m},k_{m}),B_{m^{\prime}}(v_{m^{\prime}},w_{m^{\prime}},k_{m^{\prime}})\right]\\
&\notag\quad=\sum_{x=1}^{r}S(r,x)x!\\
&\notag\quad\quad\times\sum_{\mathbf{c}^{(x)}\in\mathcal{C}_{n,x}}\mathbb{E}\left[\prod_{t=1}^{x}\mathbbm{1}_{\{M_{c_{t},m}+M_{c_{t},m^{\prime}}=0\}}\,|\,A_{m}(\mathbf{n}_{m},\mathbf{m}_{m},k_{m}),B_{m^{\prime}}(v_{m^{\prime}},w_{m^{\prime}},k_{m^{\prime}})\right]\\
&\notag\quad=\sum_{x=1}^{r}S(r,x)x!\\
&\notag\quad\quad\times\sum_{\mathbf{c}^{(x)}\in\mathcal{C}_{n,x}}\mathbb{P}[\mathbf{M}_{\mathbf{c}^{(x)},m}+\mathbf{M}_{\mathbf{c}^{(x)},m^{\prime}}=(\underbrace{0,\ldots,0}_{x})\,|\,A_{m}(\mathbf{n}_{m},\mathbf{m}_{m},k_{m}),B_{m^{\prime}}(v_{m^{\prime}},w_{m^{\prime}},k_{m^{\prime}})].
\end{align*}
Then we can use \eqref{eq:lem1} to obtain an expression for the conditional probability of $\mathbf{M}_{\mathbf{c}^{(x)},m^{\prime}}$ given $A_{m}(\mathbf{n}_{m},\mathbf{m}_{m},k_{m})$ and $B_{m^{\prime}}(v_{m^{\prime}},w_{m^{\prime}},k_{m^{\prime}})$. In particular, from \eqref{eq:lem1} we have
\begin{align}\label{r_mom}
&\mathbb{E}[(R^{\ast}_{n,m+m^{\prime}})_{[r]}\,|\,A_{m}(\mathbf{n}_{m},\mathbf{m}_{m},k_{m}),B_{m^{\prime}}(v_{m^{\prime}},w_{m^{\prime}},k_{m^{\prime}})]\\
&\notag\quad=r!\sum_{\mathbf{c}^{(r)}\in\mathcal{C}_{n,r}}\prod_{i=1}^{r}\mathbbm{1}_{\{m_{c_{i},m}=0\}}\\
&\notag\quad\quad\times\frac{(n+m-\sum_{i=1}^{k_{m}}n_{i,m}-r)_{(m^{\prime}-v_{m^{\prime}}-w_{m^{\prime}})}}{(n+m-\sum_{i=1}^{k_{m}}n_{i,m})_{(m^{\prime}-v_{m^{\prime}}-w_{m^{\prime}})}}.
\end{align}
Finally, we marginalize the last expression with respect to the distribution of the random variable $(V_{m^{\prime}},W_{m^{\prime}})\,|\,(\mathbf{N}_{m},\mathbf{M}_{m},K_{m})$. Them by combining \eqref{r_mom} with \eqref{marginal_vw}, and using the fact that \eqref{r_mom} does not depend on $K_{m^{\prime}}$, we can write 
\begin{align*}
&\notag\mathbb{E}[(R^{\ast}_{n,m+m^{\prime}})_{[r]}\,|\,A_{m}(\mathbf{n}_{m},\mathbf{m}_{m},k_{m})]\\
&\quad=r!\sum_{\mathbf{c}^{(r)}\in\mathcal{C}_{n,r}}\frac{(\theta+n+m-r)_{(m^{\prime})}}{(\theta+n+m)_{(m^{\prime})}}\prod_{i=1}^{r}\mathbbm{1}_{\{m_{c_{i},m}=0\}},
\end{align*}
and
\begin{align}\label{rth_fac_mom}
&\notag\mathbb{E}[(R_{n,m+m^{\prime}})_{[r]}\,|\,A_{m}(\mathbf{n}_{m},\mathbf{m}_{m},k_{m})]\\
&\notag\quad=\sum_{s=0}^{r}{r\choose s}(-1)^{s}(n-s)_{[r-s]} s!\sum_{\mathbf{c}^{(s)}\in\mathcal{C}_{n,s}}\frac{(\theta+n+m-s)_{(m^{\prime})}}{(\theta+n+m)_{(m^{\prime})}}\prod_{i=1}^{s}\mathbbm{1}_{\{m_{c_{i},m}=0\}}\\
&\notag\quad=\sum_{s=0}^{r}{r\choose s}(-1)^{s}(n-s)_{[r-s]} s!\frac{(\theta+n+m-s)_{(m^{\prime})}}{(\theta+n+m)_{(m^{\prime})}}\sum_{\mathbf{c}^{(s)}\in\mathcal{C}_{n,s}}\prod_{i=1}^{s}\mathbbm{1}_{\{m_{c_{i},m}=0\}}\\
&\notag\quad=\sum_{s=0}^{r}{r\choose s}(-1)^{s}(n-s)_{[r-s]}\frac{(\theta+n+m-s)_{(m^{\prime})}}{(\theta+n+m)_{(m^{\prime})}}(n-r_{m})_{[s]}\\
&\quad=\frac{r!}{(\theta+n+m)_{(m^{\prime})}}\sum_{s=0}^{r}{n-s\choose r-s}(-1)^{s}{n-r_{m}\choose s}(\theta+n+m-s)_{(m^{\prime})},
\end{align}
where $r_{m}=\sum_{1\leq i\leq n}\mathbbm{1}_{\{m_{i,m}>0\}}$. The distribution of $R_{n,m+m^{\prime}}\,|\,(\mathbf{N}_{m},\mathbf{M}_{m},K_{m})$ follows from \eqref{rth_fac_mom}. In particular, for any $x=r_{m},\ldots,\min(n,m^{\prime}+r_{m})$, we can write
\begin{align}\label{eq:post_dist}
&\notag\mathbb{P}[R_{n,m+m^{\prime}}=x\,|\,A_{m}(\mathbf{n}_{m},\mathbf{m}_{m},k_{m})]\\
&\notag=\mathbb{P}[R_{n,m+m^{\prime}}=x\,|\,R_{n,m}=r_{m}]\\
&\notag=\sum_{l\geq0}\frac{(-1)^{l}}{x!l!}\mathbb{E}[(R_{n,m+m^{\prime}})_{[x+l]}\,|\,R_{n,m}=r_{m}]\\
&\notag=\frac{1}{(\theta+n+m)_{(m^{\prime})}}\sum_{l\geq x}\frac{1}{x!}(-1)^{l-x}(l)_{[x]}\\
&\notag\quad\times\sum_{s=0}^{l}{n-s\choose l-s}(-1)^{s}{n-r_{m}\choose s}(\theta+n+m-s)_{(m^{\prime})}\\
&\notag=\frac{1}{(\theta+n+m)_{(m^{\prime})}}\sum_{s=0}^{n}(-1)^{s}{n-r_{m}\choose s}(\theta+n+m-s)_{(m^{\prime})}\\
&\notag\quad\times\sum_{l=s}^{n}(-1)^{l-x}{l\choose x}{n-s\choose l-s}\\
&\notag=\frac{(-1)^{-x}}{(\theta+n+m)_{(m^{\prime})}}\sum_{s=0}^{n}(-1)^{n-s}{n-r_{m}\choose s}{s\choose n-x}(\theta+n+m-s)_{(m^{\prime})}\\
&\notag=\frac{(-1)^{n}}{(\theta+n+m)_{(m^{\prime})}}\sum_{s=n}^{n+x}(-1)^{s}{n-r_{m}\choose s-x}{s-x\choose n-x}(\theta+n+m-s+x)_{(m^{\prime})}\\
&=\frac{{n-r_{m}\choose n-x}}{(\theta+n+m)_{(m^{\prime})}}\sum_{s=0}^{x-r_{m}}(-1)^{s}{x-r_{m}\choose s}(\theta+m-s+x)_{(m^{\prime})}.
\end{align}
The expression \eqref{eq:post_dist} coincides, after applying the Vandermonde identity, to the conditional distribution of $R_{n,m+m^{\prime}}$ given $R_{n,m}$ in \eqref{eq_post_r}, and the proof is completed.\hfill \qed

\smallskip

\textsc{Proof of ii) Theorem \ref{prp1}}. Similarly to Part i), we compute the $r$-th descending factorial moment of the random variable $R_{n,m+m^{\prime}}$ given $(\mathbf{N}_{m},\mathbf{M}_{m},K_{m})$. As a first step, we observe that we can rewrite Equation \eqref{eq_r_enlarge_freq} in the following way
\begin{align*}
\tilde{R}_{l,n,m^{\prime}}&=\sum_{i=1}^{n}(1-\mathbbm{1}_{\{M_{i,m^{\prime}}=0\}})\mathbbm{1}_{\{M_{i,m}=l\}}
=R_{l,n,m}-\sum_{i=1}^{n}\mathbbm{1}_{\{M_{i,m^{\prime}}=0\}}\mathbbm{1}_{\{M_{i,m}=l\}},
\end{align*}
and 
\begin{align*}
&\mathbb{E}[(\tilde{R}_{l,n,m^{\prime}})_{[r]}\,|\,A_{m}(\mathbf{n}_{m},\mathbf{m}_{m},k_{m}),B_{m^{\prime}}(v_{m^{\prime}},w_{m^{\prime}},k_{m^{\prime}})]\\
&\quad=\sum_{s=0}^{r}{r\choose s}(-1)^{s}(r_{l,m}-s)_{[r-s]}\mathbb{E}[(\tilde{R}^{\ast}_{l,m,m^{\prime}})_{[s]}\,|\,A_{m}(\mathbf{n}_{m},\mathbf{m}_{m},k_{m}),B_{m^{\prime}}(v_{m^{\prime}},w_{m^{\prime}},k_{m^{\prime}})].
\end{align*}
where $r_{l,m}=\sum_{1\leq i\leq n}\mathbbm{1}_{\{m_{i,m}=l\}}$ and $\tilde{R}^{\ast}_{l,m,m^{\prime}}=\sum_{1\leq i \leq n}\mathbbm{1}_{\{M_{i,m^{\prime}}=0\}}\mathbbm{1}_{\{M_{i,m}=l\}}$. By a repeated application of the Binomial theorem we can write the following expression
\begin{align*}
&\mathbb{E}[(\tilde{R}^{\ast}_{l,m,m^{\prime}})^{r}\,|\,A_{m}(\mathbf{n}_{m},\mathbf{m}_{m},k_{m}),B_{m^{\prime}}(v_{m^{\prime}},w_{m^{\prime}},k_{m^{\prime}})]\\
&\quad=\sum_{x=1}^{n}\sum_{i_{1}=1}^{r-1}\sum_{i_{2}=1}^{i_{1}-1}\cdots\sum_{i_{x-1}=1}^{i_{x-2}-1}{r\choose i_{1}}{i_{1}\choose i_{2}}\cdots{i_{x-2}\choose i_{x-1}}\\
&\notag\quad\quad\times\sum_{\mathbf{c}^{(x)}\in\mathcal{C}_{n,x}}\mathbb{E}\left[\prod_{t=1}^{x}(\mathbbm{1}_{\{M_{c_{t},m^{\prime}}=0\}}\mathbbm{1}_{\{M_{c_{t},m}=l\}})^{i_{x-t}-i_{x-t+1}}\,|\,A_{m}(\mathbf{n}_{m},\mathbf{m}_{m},k_{m}),B_{m^{\prime}}(v_{m^{\prime}},w_{m^{\prime}},k_{m^{\prime}})\right]\\
&\notag\quad=\sum_{x=1}^{r}S(r,x)x!\\
&\notag\quad\quad\times\sum_{\mathbf{c}^{(x)}\in\mathcal{C}_{n,x}}\mathbb{E}\left[\prod_{t=1}^{x}\mathbbm{1}_{\{M_{c_{t},m^{\prime}}=0\}}\mathbbm{1}_{\{M_{c_{t},m}=l\}}\,|\,A_{m}(\mathbf{n}_{m},\mathbf{m}_{m},k_{m}),B_{m^{\prime}}(v_{m^{\prime}},w_{m^{\prime}},k_{m^{\prime}})\right]\\
&\notag\quad=\sum_{x=1}^{r}S(r,x)x!\\
&\notag\quad\quad\times\sum_{\mathbf{c}^{(x)}\in\mathcal{C}_{n,x}}\prod_{t=1}^{x}\mathbbm{1}_{\{m_{c_{t},m}=l\}}\mathbb{E}\left[\prod_{t=1}^{x}\mathbbm{1}_{\{M_{c_{t},m^{\prime}}=0\}}\,|\,A_{m}(\mathbf{n}_{m},\mathbf{m}_{m},k_{m}),B_{m^{\prime}}(v_{m^{\prime}},w_{m^{\prime}},k_{m^{\prime}})\right]\\
&\notag\quad=\sum_{x=1}^{r}S(r,x)x!\\
&\notag\quad\quad\times\sum_{\mathbf{c}^{(x)}\in\mathcal{C}_{n,x}}\prod_{t=1}^{x}\mathbbm{1}_{\{m_{c_{t},m}=l\}}\mathbb{P}[\mathbf{M}_{\mathbf{c}^{(x)},m^{\prime}}=(\underbrace{0,\ldots,0}_{x}) \,|\,A_{m}(\mathbf{n}_{m},\mathbf{m}_{m},k_{m}),B_{m^{\prime}}(v_{m^{\prime}},w_{m^{\prime}},k_{m^{\prime}})]\\
&\notag\quad=\sum_{x=1}^{r}S(r,x)x!\\
&\notag\quad\quad\times\sum_{\mathbf{c}^{(x)}\in\mathcal{C}_{n,x}}\prod_{t=1}^{x}\mathbbm{1}_{\{m_{c_{t},m}=l\}}\frac{(n+m-\sum_{i=1}^{k_{m}}n_{i,m}-\sum_{i=1}^{x}(1+m_{c_{i},m}))_{(m^{\prime}-v_{m^{\prime}}-w_{m^{\prime}})}}{(n+m-\sum_{i=1}^{k_{m}}n_{i,m})_{(m^{\prime}-v_{m^{\prime}}-w_{m^{\prime}})}},
\end{align*}
i.e.,
\begin{align}\label{r_mom_freq}
&\mathbb{E}[(\tilde{R}^{\ast}_{l,m,m^{\prime}})_{[r]}\,|\,A_{m}(\mathbf{n}_{m},\mathbf{m}_{m},k_{m}),B_{m^{\prime}}(v_{m^{\prime}},w_{m^{\prime}},k_{m^{\prime}})]\\
&\notag\quad=r!\sum_{\mathbf{c}^{(r)}\in\mathcal{C}_{n,r}}\prod_{t=1}^{r}\mathbbm{1}_{\{m_{c_{t},m}=l\}}\frac{(n+m-\sum_{i=1}^{k_{m}}n_{i,m}-\sum_{i=1}^{r}(1+m_{c_{i},m}))_{(m^{\prime}-v_{m^{\prime}}-w_{m^{\prime}})}}{(n+m-\sum_{i=1}^{k_{m}}n_{i,m})_{(m^{\prime}-v_{m^{\prime}}-w_{m^{\prime}})}}.
\end{align}
Finally, we marginalize the last expression with respect to the distribution of the random variable $(V_{m^{\prime}},W_{m^{\prime}})\,|\,(\mathbf{N}_{m},\mathbf{M}_{m},K_{m})$. By combining \eqref{r_mom_freq} with \eqref{marginal_vw} one has
\begin{align*}
&\mathbb{E}[(\tilde{R}^{\ast}_{l,m,m^{\prime}})_{[r]}\,|\,A_{m}(\mathbf{n}_{m},\mathbf{m}_{m},k_{m})]\\
&\quad=r!\sum_{\mathbf{c}^{(r)}\in\mathcal{C}_{n,r}}\prod_{t=1}^{r}\mathbbm{1}_{\{m_{c_{t},m}=l\}}\frac{(\theta+n+m-\sum_{i=1}^{r}(1+m_{c_{i},m}))_{(m^{\prime})}}{(\theta+n+m)_{(m^{\prime})}}\\
&\quad=r!{r_{l,m}\choose r}\frac{(\theta+n+m-r(1+l))_{(m^{\prime})}}{(\theta+n+m)_{(m^{\prime})}}
\end{align*}
and
\begin{align}\label{rth_fac_mom_freq}
&\mathbb{E}[(\tilde{R}_{l,m,m^{\prime}})_{[r]}\,|\,A_{m}(\mathbf{n}_{m},\mathbf{m}_{m},k_{m})]\\
&\notag\quad=\sum_{s=0}^{r}{r\choose s}(-1)^{s}(r_{l,m}-s)_{[r-s]}s!{r_{l,m}\choose s}\frac{(\theta+n+m-s(1+l))_{(m^{\prime})}}{(\theta+n+m)_{(m^{\prime})}}.
\end{align}
Accordingly, the distribution of $R_{n,m+m^{\prime}}\,|\,(\mathbf{N}_{m},\mathbf{M}_{m},K_{m})$ follows from \eqref{rth_fac_mom_freq}. In particular, for any $x=0,\ldots,\min(r_{l,m},m^{\prime})$, we can write the following expression
\begin{align*}
&\mathbb{P}[\tilde{R}_{l,m,m^{\prime}}=x\,|\,A_{m}(\mathbf{n}_{m},\mathbf{m}_{m},k_{m})]\\
&\quad=\mathbb{P}[\tilde{R}_{l,m,m^{\prime}}=x\,|\,R_{l,n,m}=r_{l,m}]\\
&\quad=\sum_{y\geq0}\frac{(-1)^{y}}{x!y!}\mathbb{E}[(\tilde{R}_{l,m,m^{\prime}})_{[x+l]}\,|\,R_{l,n,m}=r_{l,m}]\\
&\quad=\sum_{y\geq0}(-1)^{y}\frac{1}{x!y!}\sum_{s=0}^{x+y}{x+y\choose s}(-1)^{s}(r_{l,m}-s)_{[x+y-s]}\\
&\quad\quad\times s!{r_{l,m}\choose s}\frac{(\theta+n+m-s(1+l))_{(m^{\prime})}}{(\theta+n+m)_{(m^{\prime})}}\\
&\quad=\sum_{y\geq x}(-1)^{y-x}\frac{1}{x!(y-x)!}\sum_{s=0}^{y}{y\choose s}(-1)^{s}(r_{l,m}-s)_{[y-s]}\\
&\quad\quad\times s!{r_{l,m}\choose s}\frac{(\theta+n+m-s(1+l))_{(m^{\prime})}}{(\theta+n+m)_{(m^{\prime})}}\\
&\quad=\sum_{y=0}^{r_{l,m}}(-1)^{y-x}{y\choose x}\sum_{i=0}^{y}{r_{l,m}-s\choose y-s}(-1)^{s}{r_{l,m}\choose s}\frac{(\theta+n+m-s(1+l))_{(m^{\prime})}}{(\theta+n+m)_{(m^{\prime})}}\\
&\quad=\sum_{s=0}^{r_{l,m}}(-1)^{s-x}{r_{l,m}\choose s}\frac{(\theta+n+m-s(1+l))_{(m^{\prime})}}{(\theta+n+m)_{(m^{\prime})}}\sum_{y=s}^{r_{l,m}}(-1)^{y}{y\choose x}{r_{l,m}-s\choose y-s}\\
&\quad=\sum_{s=0}^{r_{l,m}}(-1)^{-x}{r_{l,m}\choose s}\frac{(\theta+n+m-s(1+l))_{(m^{\prime})}}{(\theta+n+m)_{(m^{\prime})}}\\
&\quad\quad\times\sum_{y=0}^{r_{l,m}-s}(-1)^{y}{y+s\choose x}{r_{l,m}-s\choose y}\\
&\quad=\sum_{s=0}^{r_{l,m}}(-1)^{-x}{r_{l,m}\choose s}\frac{(\theta+n+m-s(1+l))_{(m^{\prime})}}{(\theta+n+m)_{(m^{\prime})}}(-1)^{r_{l,m}-s}{s	\choose x-r_{l,m}+s}.
\end{align*}
\hfill \qed

\smallskip

\textsc{Proof of Proposition \ref{cor1}}. Recalling the definitions of the random variables $R_{n,m}$ and $R_{n,m+m^{\prime}}$, let $\tilde{R}_{n,m^{\prime}}=R_{n,m+m^{\prime}}-R_{n,m}$, that is the number of distinct types in the additional sample $\mathbf{X}_{m^{\prime}}$ that coincide with the atoms $Z_{i}^{\ast}$ that are not in the initial sample $\mathbf{X}_{m}$. In other terms $\tilde{R}_{n,m^{\prime}}$ denotes the number of new types induced by $\mathbf{X}_{m^{\prime}}$ that coincide with the atoms $Z_{i}^{\ast}$. From \eqref{eq_post_r}, we can write
\begin{equation}\label{eq_turing_1}
\mathbb{P}[\tilde{R}_{n,1}=1\,|\,R_{n,m}=y]=\mathbb{E}[\tilde{R}_{n,1}\,|\,R_{n,m}=y]=\frac{n-y}{\theta+n+m}.
\end{equation}
See also the factorial moment formula \eqref{rth_fac_mom} with $r=1$ and $m^{\prime}=1$. Also, from \eqref{eq_post_r_freq}, 
\begin{align}\label{eq_turing_2}
\mathbb{P}[\tilde{R}_{l,n,1}=1\,|\,R_{l,n,m}=y]&=\mathbb{E}[\tilde{R}_{l,n,1}\,|\,R_{l,n,m}=y]\\
&\notag=y\left(1-\frac{\theta+n+m-(1+l)}{\theta+n+m}\right).
\end{align}
See also the factorial moment formula in Equation \eqref{rth_fac_mom_freq} with $r=1$ and $m^{\prime}=1$. The proof is completed by simply randomizing the parameter $n$ appearing in  \eqref{eq_turing_1} and in \eqref{eq_turing_2} with respect to the distribution \eqref{eq_bayes} and \eqref{eq_bayes_freq}, respectively.\hfill \qed


\section*{Acknowledgements}
We are grateful to two anonymous referees, whose comments and suggestions have helped us to improve the paper substantially. S. Favaro and S. Feng acknowledge the hospitality of the Banff International Research Station, Canada, where this project has been completed during the Research in Team Programme ``Random partitions and Bayesian nonparametrics''. S. Favaro is supported by the European Research Council through StG N-BNP 306406. Shui Feng is supported by the Natural Sciences and Engineering Research Council of Canada. P. Jenkins is supported in part by Engineering and Physical Sciences Research Council Grant EP/L018497/1.



\end{document}